\documentclass[a4paper,10pt,reqno]{amsart}
\usepackage[centertags]{amsmath}
\usepackage{amsfonts,amssymb,amsthm}
\usepackage{float,longtable}
\usepackage{graphicx}
\usepackage{a4wide}
\usepackage[latin1]{inputenc}
\usepackage{newlfont}
 \theoremstyle{plain}
 \begingroup
 \newtheorem{thm}{Theorem}[section]

 \endgroup

 \theoremstyle{definition}
 \begingroup

 \endgroup

 \theoremstyle{remark}
 \begingroup
 \newtheorem{rem}[thm]{Remark}
 \endgroup

 \numberwithin{equation}{section}

\newcommand{\norm}[1]{\left\Vert#1\right\Vert}
\newcommand{\snorm}[1]{\left|#1\right|}

\title[3D-1D asymptotic analysis for thin curved domains in nonlinear elasticity]
{The nonlinear bending-torsion theory for curved rods as
$\Gamma$-limit of three-dimensional elasticity }
\author[]{LUCIA SCARDIA}
\address[]{S.I.S.S.A., Via Beirut 2-4, 34014, Trieste, Italy}
\email[]{scardia@sissa.it}

\begin{document}
\maketitle
\begin{center}
\begin{minipage}{12cm}
\footnotesize{ \noindent {\bf Abstract.}
The problem of the rigorous derivation of one-dimensional models for nonlinearly elastic curved beams is studied in a variational setting. Considering different scalings of the three-dimensional energy and 
passing to the limit as the diameter of the beam goes to zero, a nonlinear model for strings and a bending-torsion theory for rods are deduced.

\vspace{15pt}

\vspace{6pt}
}
\end{minipage}
\end{center}

\bigskip

\section{Introduction}

This paper is part of a series of recent works concerning the rigorous derivation of lower 
dimensional models for thin domains from nonlinear three-dimensional elasticity, 
by means of $\Gamma-$~convergence. 

The first result in this direction is due to E. Acerbi, G. Buttazzo and D. Percivale (see \cite{ABP91}), who
deduced a nonlinear model for elastic strings by means of a 3D-1D dimension reduction. 
The two-dimensional analogue was studied by H. Le Dret and A. Raoult in \cite{LDR95}, where they derived a nonlinear model for elastic membranes. The more delicate case of plates was justified more recently by G. Friesecke, R.D. James and S. M\"uller in \cite{FJM02} (see also \cite{FJMP} for a complete survey on plate theories). The case of shells was considered in \cite{LDR00} and \cite{FJMM03}.

As for one-dimensional models, nonlinear theories for elastic rods have been deduced by M.G. Mora, S. M\"uller (see \cite{MM03}, \cite{MGMM04}) and, independently, by O. Pantz (see \cite{P02}). 
In all these results, as in \cite{ABP91}, the beam is assumed to be straight in the unstressed configuration.

In this paper we study the case of a heterogeneous curved beam made of a hyperelastic material. 
Let $\Omega:=(0,L)\times D$, where $L>0$ and $D$ is a bounded domain in $\mathbb{R}^2$, and let $h>0$.
We consider a beam, whose reference configuration is given by
$$\widetilde{\Omega}_{h}:= \{\gamma(s) + h\,\xi\,\nu_{2}(s) + h\,\zeta\, \nu_{3}(s) : (s,\xi,\zeta)\in \Omega\},$$
where $\gamma:(0,L)\to \mathbb{R}^{3}$ is a smooth simple curve describing the mid-fiber of the beam, and $\nu_2, \nu_3:(0,L)\to \mathbb{R}^{3}$ are two smooth vectors such that $(\gamma',\nu_2, \nu_3)$ provide an orthonormal frame along the curve. In particular, the cross section of the beam is constant along $\gamma$ and is given by the set $hD$.
It is natural to parametrize $\widetilde{\Omega}_{h}$ through the map 
\begin{equation*}
\Psi^{(h)} :  \Omega \rightarrow \widetilde{\Omega}_{h}, \quad (s,\xi,\zeta)\mapsto \gamma(s) + h\,\xi\,\nu_{2}(s) + h\,\zeta\,\nu_{3}(s),
\end{equation*}
which is one-to-one for $h$ small enough.

The starting point of our approach is the elastic energy per unit volume
\begin{equation*}\label{puv}
\tilde{I}^{(h)}(\tilde{v}):= \frac{1}{h^2}\int_{\widetilde{\Omega}_{h}}
W\big(\big(\Psi^{(h)}\big)^{-1} (x),\nabla \tilde{v}(x)\big) dx
\end{equation*}
of a deformation $\tilde{v} \in W^{1,2}(\widetilde{\Omega}_{h};\mathbb{R}^{3})$. The stored energy density $W:\Omega\times {\mathbb M}^{3\times 3}\to [0,+\infty]$ has to satisfy some natural conditions; i.e.,
\smallskip
\begin{itemize}
\item $W$ is frame indifferent: 
$W(z,RF) = W(z,F)$ for a.e.\ 
$z\,\in \Omega$, every $F\in \mathbb{M}^{3\times 3}$, and every $R\in SO(3)$;
\smallskip
\item $W(z,F)\geq C\,\mbox{dist}^{2}(F,SO(3))$  for a.e.\  $z\in \Omega$ and
every $F\in\mathbb{M}^{3\times 3}$;
\smallskip
\item $W(z,R)=0$ for a.e.\ $z\in\Omega$ and every $R\in SO(3)$. 
\end{itemize}
\smallskip
For the complete list of assumptions on $W$ we refer to Section~2.

The aim of this work is to study the asymptotic behaviour of different scalings of the energy $\tilde{I}^{(h)}$, as $h\rightarrow 0$, by means of $\Gamma$-convergence  (see \cite{DM93} for a comprehensive
introduction to $\Gamma$-convergence). 
Heuristic arguments suggest that, as in the case of straight beams, energies of order $1$ correspond to stretching and shearing deformations, leading to a \textit{string theory}, while energies of order $h^2$ correspond to bending flexures and torsions keeping the mid-fiber unextended, leading to a \textit{rod theory}.

The main results of the paper are contained in Section $3$, where we identify the $\Gamma$-limit of the sequence of functionals $\big(\tilde{I}^{(h)}/h^2\big)$. We first show a compactness result for sequences of deformations having equibounded energies (Theorem \ref{compactness}). More precisely, given a sequence $\big(\tilde{v}^{(h)}\big)\subset W^{1,2}(\widetilde{\Omega}_{h};\mathbb{R}^{3})$ with $\tilde{I}^{(h)}(\tilde{v}^{(h)})/h^2\leq C$, we prove that there exist a subsequence (not relabelled) and some
constants $c^{(h)}\in\mathbb{R}^3$  such that 
\begin{align*}
\tilde{v}^{(h)}\circ \Psi^{(h)} -
c^{(h)} &\rightarrow v \quad \mbox{strongly in}\,\,
W^{1,2}(\Omega;\mathbb{R}^{3}), 
\\
\frac{1}{h}\,\partial_\xi\big(\tilde{v}^{(h)}\circ
\Psi^{(h)}\big)&\,\rightarrow d_2 \quad \mbox{strongly in}\,
L^{2}(\Omega;\mathbb{R}^{3}),
\\
\frac{1}{h}\,\partial_\zeta\big(\tilde{v}^{(h)}\circ
\Psi^{(h)}\big)\,&\rightarrow d_3 \quad \mbox{strongly in}\,
L^{2}(\Omega;\mathbb{R}^{3}),
\end{align*}
where $(v,d_{2},d_{3})$ belongs to the class 
\begin{eqnarray*}
\mathcal{A}:= \{(v,d_{2},d_{3}) \in W^{2,2}((0, L);\mathbb{R}^{3})\times W^{1,2}((0, L);\mathbb{R}^{3})\times 
W^{1,2}((0, L);\mathbb{R}^{3}):\nonumber\\
(v'(s)\,|\,d_{2}(s)\,|\,d_{3}(s))\in SO(3)\,\,\mbox{for a.e.\ } s \mbox{ in}\,\,  (0, L)\}.
\end{eqnarray*} 
The key ingredient in the proof is a geometric rigidity theorem proved by G. Friesecke, R.D. James and S. M\"uller in \cite{FJM02}.
In Theorems \ref{scithm} and \ref{bfa} we show that the $\Gamma$-limit of the sequence $\big(\tilde{I}^{(h)}/h^2\big)$ is given by
\begin{equation}\label{introd}
I(v,d_2,d_3):=
\left\{
\vspace{.7cm}
\begin{array}{ll}
\displaystyle\frac{1}{2}\int_{0}^{L} Q_{2} \big(s,\big(R^{T}(s)R'(s) -
R_{0}^{T}(s)R_{0}'(s)\big)\big)ds  & \mbox{if } \, (v,d_2,d_3)\in \mathcal{A},\\
\displaystyle + \infty & \mbox{otherwise}, 
\end{array}
\right.
\end{equation}
where $R := (v'\,|\,d_2\,|\,d_3)$, $R_{0}: = (\gamma'\,|\,\nu_2\,|\,\nu_3)$,  and $Q_{2}$ is a quadratic form arising from a minimization procedure involving the quadratic form of linearized elasticity (see (\ref{Q2})).
We point out that in Theorems \ref{scithm} and \ref{bfa} we do not require any growth condition from above on the energy density~$W$.

We notice that in the limit problem the behaviour of the rod is described by a triple $(v,d_2,d_3)$. The function $v$ represents the deformation of the mid-fiber, which satisfies $|v'|=1$ a.e., because of the constraint $(v'\,|\,d_{2}\,|\,d_{3})\in SO(3)$ a.e.. Therefore, the admissible deformations are only those leaving the mid-fiber unextended. Moreover, the triple $(v,d_2,d_3)$ provides an orthonormal frame along the deformed curve; in particular, $d_2$ and $d_3$ belong to the normal plane to the deformed curve and describe the rotation undergone by the cross section.

Since $R = (v'\,|\,d_2\,|\,d_3)$ is a rotation a.e., the matrix $R^TR'$ is skew-symmetric a.e.\ and its entries are given by 
$$
(R^TR')_{1k}=-(R^TR')_{k1}=v'\cdot d_k' \quad \hbox{for }k=2,3,\quad
(R^TR')_{23}=-(R^TR')_{32}=d_2\cdot d_{3}'.
$$
It is easy to see that the scalar products $v'\cdot d_k'$ are related to curvature and therefore, to bending effects, while $d_2\cdot d_{3}'$ is related to torsion and twist.
We remark also that the energy depends explicitly on the reference state of the beam through the quantity $R_0^{T}R'_0$, which encodes informations about the bending and torsion of the beam in the initial configuration.

We notice that, specifying $R_{0} = Id$ in (\ref{introd}), we recover the result for straight rods obtained in \cite{MM03} and \cite{P02}.

The last section of the paper is devoted to the study of lower scalings of the energy. Assuming that 
the energy density $W$ satisfies a growth condition from above, we prove the $\Gamma$-convergence of the sequence $\big(\tilde{I}^{(h)}\big)$  to a functional corresponding to a string model.  Finally we show that intermediate scalings of the energy between $1$ and $h^2$ lead to a trivial $\Gamma$-limit.

\section{Notations and formulation of the problem}

In this section we describe the geometry of the unstressed curved beam. 
Let $\gamma : [0, L] \rightarrow \mathbb{R}^{3}$ be a simple regular curve of class $C^{2}$ parametrized by the arc-length and let 
$\tau = \dot{\gamma}$ be its unitary tangent vector. We assume that there exists an orthonormal frame of class $C^{1}$ along the curve. More precisely, we assume that there exists $R_{0}\in C^{1}([0, L]; \mathbb{M}^{3\times 3})$ such that $R_{0}(s)\in SO(3)$ for every $s \in [0, L]$ and $R_{0}(s)\,e_{1} = \tau(s)$ for every $s \in [0, L]$, where $e_i$, for $i=1,2,3$, denotes the i-th vector of the canonical basis of $\mathbb{R}^{3}$ and \mbox{$SO(3) = \big\{R\in\mathbb{M}^{3\times 3} : R^T R = Id,\, \det R = 1 \big\}$}.
We set $\nu_k (s) := R_{0}(s)\,e_{k}$ for $k = 2,3$.

Let $D\subset \mathbb{R}^{2}$ be a bounded open connected set with Lipschitz boundary such that
\begin{equation}\label{dom1}
\int_{D}\xi\,\zeta\, d\xi\, d\zeta = 0
\end{equation} 
and
\begin{equation}\label{dom2}
\int_{D}\xi\,d\xi\,d\zeta = \int_{D}\zeta\,d\xi\,d\zeta = 0,
\end{equation} 
where $(\xi,\zeta)$ stands for the coordinates of a generic point of $D$. Without loss of generality, we can assume $\mathcal{L}^2(D) = 1$. We set $\Omega:= (0, L)\times D$.
\newline
The reference configuration of the thin beam is given by
$$\widetilde{\Omega}_{h}:= \{\gamma(s) + h\,\xi\,\nu_{2}(s) + h\,\zeta\, \nu_{3}(s) : (s,\xi,\zeta)\in \,\Omega\},$$
where $h$ is a small positive parameter. Clearly the curve $\gamma$ and the set $D$ represent the middle fiber and the cross section of the beam, respectively.
The set $\widetilde{\Omega}_{h}$ is parametrized by the map   
\begin{equation*}
\Psi^{(h)} : \Omega \rightarrow \widetilde{\Omega}_{h}\,: \quad (s,\xi,\zeta)\mapsto \gamma(s) + h\,\xi\,\nu_{2}(s) + h\,\zeta\,\nu_{3}(s),
\end{equation*}
which is one-to-one for $h$ small enough and of class $C^1$.

We assume that the thin beam is made of a hyperelastic material whose stored energy density 
$W : \Omega\times\mathbb{M}^{3\times 3} \rightarrow [0, + \infty]$ is a Carath\'eodory function
satisfying the following hypotheses:
\begin{itemize}
\vspace{.2cm}
\item[(i)] there exists $\delta>0$ such that the function $F\mapsto W(z,F)$ is of class $C^{2}$ on the set\\$\big\{F\in\mathbb{M}^{3\times 3}: \mbox{dist}(F,SO(3)) < \delta\big\}$  for a.e. $z\,\in \Omega$;
\vspace{.2cm}
\item[(ii)] the second derivative $\partial^{2}W/\partial F^{2}$ is a Carath\'eodory function on the set
\begin{equation}\label{set}
\Omega\times\{F\in \mathbb{M}^{3\times 3}:\,\mbox{dist}(F,SO(3)) < \delta \}
\end{equation}
and there exists a constant $C_{1} > 0$ such that
\begin{align*}
&\bigg|\frac{\partial^{2}W}{\partial F^{2}}(z,F)[G,G]\bigg| \leq C_{1} |\,G\,|^{2}\quad \mbox{for a.e. } z\in \Omega,\, \mbox{every} F \hbox{with} \hbox{dist}(F,SO(3))<\delta\\
& \mbox{and every } G\in \mathbb{M}^{3\times 3}_{sym};
\end{align*}
\item[(iii)] $W$ is frame indifferent, i.e., $W(z,RF) = W(z,F)$  for a.e. $z\,\in \Omega$, every $F\in \mathbb{M}^{3\times 3}$ and every $R\in SO(3)$;
\vspace{.2cm}
\item[(iv)] $W(z,R)=0$ for every  $R\in SO(3)$;
\vspace{.2cm}
\item[(v)] $\exists$ $C_{2} >\,0$ independent of $z$ such that $W(z,F)\geq C_{2}\,
\mbox{dist}^{2}(F,SO(3))$  for a.e.  $z\in \Omega$ and every  $F\in\mathbb{M}^{3\times 3}$. 
\end{itemize}
Notice that, since we do not require any growth condition from above, $W$ is allowed to assume the value $+ \infty$ outside the set (\ref{set}). Therefore our treatment covers the physically relevant case in which $W = + \infty$ for $\det F < 0$, $W\rightarrow + \infty$ as $\det F \rightarrow 0^+$.

Let $\tilde{v} \in W^{1,2}(\widetilde{\Omega}_{h};\mathbb{R}^{3})$ be a deformation of $\widetilde{\Omega}_{h}$. 
The elastic energy per unit volume associated to $\tilde{v}$ is defined by
\begin{equation}\label{energypuv}
\tilde{I}^{(h)}(\tilde{v}):= \frac{1}{h^{2}}\int_{\widetilde{\Omega}_{h}}
W\big(\big(\Psi^{(h)}\big)^{-1} (x),\nabla
\tilde{v}(x)\big) dx.
\end{equation} 
The main part of this work is devoted to the study of the asymptotic behaviour as $h\rightarrow 0$ of the sequence of functionals $\tilde{I}^{(h)}/h^{2}$.
In the final part we will also discuss the scaling $\tilde{I}^{(h)}/h^{\alpha}$ \,
for $0 \leq \alpha < 2$.

We conclude this section by analysing some properties of the map $\Psi^{(h)}$, which will be useful in the sequel. We will use the following notation: for any function $z\in W^{1,2}(\Omega;\mathbb{R}^3)$ we set
\begin{equation*}
\nabla_{h}z := \left(\partial_s z\,\Big|\,\frac{1}{h}\,\partial_\xi z\,\Big|\, \frac{1}{h}\,\partial_\zeta z\right).
\end{equation*}
We observe that $\nabla_{h}\Psi^{(h)}$
can be written as the sum of the rotation $R_{0}$ and a
perturbation of order $h$, that is,
$$\nabla_{h}\Psi^{(h)}(s,\xi,\zeta) = R_{0}(s) + h\,\left(\xi\,\nu'_{2}(s) + \zeta\,\nu'_{3}(s)\right)\otimes e_1. $$
From this fact it follows that, as $h\rightarrow 0$,
\begin{equation}\label{convdet}
\nabla_{h}\Psi^{(h)}(s,\xi,\zeta)\rightarrow  R_{0}(s)\quad \mbox{and}\quad\det \big(\nabla_{h}\Psi^{(h)}\big) \rightarrow \det R_{0} = 1 \,\,\,\mbox{uniformly}.
\end{equation}
This implies that for $h$ small enough
$\nabla_{h}\Psi^{(h)}$ is invertible at each
point of $\Omega$. Since the inverse of
$\nabla_{h}\Psi^{(h)}$ can be written as
\begin{equation}\label{invA}
\big(\nabla_{h}\Psi^{(h)}\big)^{-1}(s,\xi,\zeta) =
R_{0}^{T}(s) - h\,R_{0}^{T}(s)\,\big[(\xi\,\nu'_{2}(s) +
\zeta\,\nu'_{3}(s))\otimes e_1\big] R_{0}^{T}(s) + O(h^{2})
\end{equation}
with \, $O(h^{2})/h^{2}$ \, uniformly bounded, we have also that 
$\big(\nabla_{h}\Psi^{(h)}\big)^{-1}$ converges
to $R_{0}^{T}$ uniformly.

\section{Derivation of the bending-torsion theory for curved rods}
\noindent
The aim of this section is the study of the asymptotic behaviour of the sequence of functionals
\begin{equation*}
\frac{1}{h^{2}}\,\tilde{I}^{(h)}(\tilde{v})= \frac{1}{h^{4}} \int_{\widetilde{\Omega}_{h}} W\big(\big(\Psi^{(h)}\big)^{-1} (x),\nabla \tilde{v}(x)\big) dx
\end{equation*}
under the assumptions (i)-(v) of Section 2. 

\subsection{Compactness}
We will show a compactness result for sequences of deformations having equibounded energy \, $\tilde{I}^{(h)}/h^2$. 
A key ingredient in the proof is the following rigidity result, proved
by G. Friesecke, R.D. James and S. M\"{u}ller in \cite{FJM02}.
\begin{thm}\label{Teorigid}
Let $U$ be a bounded Lipschitz domain in $\mathbb{R}^{n}$, $n\geq
2$. Then there exists a constant $C(U)$ with the following
property: for every $u\in W^{1,2}(U;\mathbb{R}^{n})$ there is an
associated rotation $R\in SO(n)$ such that
\begin{equation*}
\norm{\nabla u - R}_{L^{2}(U)} \leq C(U)\norm{\textnormal{dist}(\nabla
u, SO(n))}_{L^{2}(U)}.
\end{equation*}
\end{thm}
\begin{rem}
The constant $C(U)$ can be chosen independent of $U$ for a family
of sets that are Bilipschitz images of a cube (with uniform
Lipschitz constants), as remarked in \cite{FJMM03}.
\end{rem}
\noindent
We introduce the class of limiting admissible deformations
\begin{eqnarray}\label{defA}
\mathcal{A}:= \{(v,d_{2},d_{3}) \in W^{2,2}((0, L);\mathbb{R}^{3})\times W^{1,2}((0, L);\mathbb{R}^{3})\times 
W^{1,2}((0, L);\mathbb{R}^{3}):\nonumber\\
(v'(s)\,|\,d_{2}(s)\,|\,d_{3}(s))\in SO(3)\,\,\mbox{for a.e. s in}\,\,  (0, L)\}.
\end{eqnarray}
\noindent
Now we are ready to state and prove the main result of this subsection.
\begin{thm}\label{compactness}
Let $\big(\tilde{v}^{(h)}\big)$ be a sequence in
$W^{1,2}\big(\widetilde{\Omega}_{h};\mathbb{R}^{3}\big)$ such
that
\begin{equation}\label{finite}
\frac{1}{h^2}\,\tilde{I}^{(h)}(\tilde{v}^{(h)})
\leq c < +\infty.
\end{equation}
Then there exist a triple $(v,d_{2},d_{3})\in \mathcal{A}$, a map
$\overline{R}\in W^{1,2}((0, L);\mathbb{M}^{3\times 3})$ with 
$\overline{R}(s)\in SO(3)$ \\for a.e. $s\in [0, L]$, and some
constants $c^{(h)}\in\mathbb{R}^3$ such that, up to subsequences,
\begin{align}
\tilde{v}^{(h)}\circ \Psi^{(h)} -
c^{(h)} &\rightarrow v \quad \mbox{strongly in}\,\,
W^{1,2}(\Omega;\mathbb{R}^{3}), \label{teo1}\\
\frac{1}{h}\,\partial_\xi\big(\tilde{v}^{(h)}\circ
\Psi^{(h)}\big)&\,\rightarrow d_2 \quad \mbox{strongly in}\,
L^{2}(\Omega;\mathbb{R}^{3}),\label{teo2}\\
\frac{1}{h}\,\partial_\zeta\big(\tilde{v}^{(h)}\circ
\Psi^{(h)}\big)\,&\rightarrow d_3 \quad \mbox{strongly in}\,
L^{2}(\Omega;\mathbb{R}^{3}),\label{teo25}\\
\nabla\tilde{v}^{(h)}\circ
\Psi^{(h)}&\rightarrow \overline{R}\quad \mbox{strongly
in}\, L^{2}(\Omega;\mathbb{M}^{3\times 3}).\label{teo3}
\end{align}
Moreover, for a.e. $s\in [0, L]$, we have \,\, $(v'(s)\,|\,d_{2}(s)\,|\,d_{3}(s)) = \overline{R}(s)\,R_{0}(s)$, where $R_{0} = (\tau\,|\,\nu_{2}\,|\,\nu_{3})$.
\end{thm}
\begin{proof}
Let $\big(\tilde{v}^{(h)}\big)$ be a sequence in
$W^{1,2}(\widetilde{\Omega}_{h};\mathbb{R}^{3})$ satisfying (\ref{finite}). 
The assumption (v) on $W$ implies that
\begin{equation*}
\int_{\widetilde{\Omega}_{h}} \textrm{dist}^{2}\big(\nabla
\tilde{v}^{(h)}(x),SO(3)\big) dx < C\,h^4
\end{equation*}
for a suitable constant $C$. Using the change of variables $\Psi^{(h)}$, we have
\begin{equation}\label{bound2}
\int_{\Omega} \mbox{dist}^{2}\big(\nabla
\tilde{v}^{(h)}\circ\Psi^{(h)},SO(3)\big) \det
\big(\nabla_{h}\Psi^{(h)}\big)
ds\,d\xi\,d\zeta \leq c\,h^{2}.
\end{equation}
From (\ref{convdet}) and the estimate
$$\mbox{dist}^{2}(F,SO(3)) \geq \frac{1}{2}\,|\,F\,|^{2} - 3,$$
we get the bound
\begin{equation}\label{uff}
\int_{\Omega}\big|\,\nabla\tilde{v}^{(h)}\circ
\Psi^{(h)}\big|^{2} ds\,d\xi\,d\zeta  \leq c.
\end{equation}
Define the sequence $F^{(h)} :=
\nabla\tilde{v}^{(h)}\circ\Psi^{(h)}$; from
(\ref{uff}) it follows that there exists a function $F \in
L^{2}(\Omega;\mathbb{M}^{3\times 3})$ such that, up to
subsequences,
\begin{equation}\label{Fconv}
F^{(h)} \rightharpoonup F \quad \textrm{weakly in} \,
L^{2}(\Omega;\mathbb{M}^{3\times 3}).
\end{equation}
Using Theorem \ref{Teorigid}, we will show that this convergence is 
in fact strong in $L^{2}$ and that the limit function $F$ is a rotation 
a.e. depending only on the variable along the mid-fiber and
belonging to $W^{1,2}((0, L);\mathbb{M}^{3\times 3})$. The idea is
to divide the domain $\widetilde{\Omega}_{h}$ in small
curved cylinders, which are images of homotetic straight cylinders 
through the same Bilipschitz function. Then, we can apply the rigidity 
theorem to each small curved cylinder with the same
constant. In this way we construct a piecewise constant rotation, 
which is close to the deformation gradient $\nabla\tilde{v}^{(h)}$ 
in the $L^{2}$ norm. For every small enough $h>0$, let $K_{h}\in
\mathbb{N}$ satisfy
$$h \leq \frac{L}{K_{h}} < 2\,h.$$
For every $a\in[0, L)\cap
\dfrac{L}{K_{h}}\,\mathbb{N}$, define the segments
$$
S_{a,K_{h}}:=
\left\{
\begin{array}{ll}
\vspace{.2cm}
(a, a + 2\,h) & \mbox{if } \,\, a< L-\dfrac{L}{K_{h}},\\
(L - 2\,h,L ) & \mbox{otherwise}.
\end{array}
\right.
$$
Now consider the cylinders $C_{a,h}:= S_{a,K_{h}}\times D$
and the subsets of $\widetilde{\Omega}_{h}$ defined by
$\widetilde{C}_{a,h}:=
\Psi^{(h)}(C_{a,h})$. Remark that 
$\widetilde{C}_{a,h}$ is a Bilipschitz image of a cube of size $h$, that is 
$(a,0,0) + h\,\big((0, 2)\times D \big)$, through  the map $\Psi$ defined as
\begin{equation*}
\Psi :  [0, L]\times\mathbb{R}^{2} \rightarrow \mathbb{R}^{3}, \quad (s,y_{2},y_{3})\mapsto \gamma(s) + y_{2}\,\nu_{2}(s) + y_{3}\,\nu_{3}(s). 
\end{equation*}
By Theorem \ref{Teorigid} we obtain that there exists a constant rotation $\widetilde{R}_{a}^{(h)}$ such
that
\begin{equation}\label{rigid}
\int_{\widetilde{C}_{a,h}}\big|\,\nabla\tilde{v}^{(h)}
- \widetilde{R}_{a}^{(h)}\big|^{2} dx \leq c
\int_{\widetilde{C}_{a,h}}\mbox{dist}^{2}(\nabla\tilde{v}^{(h)},SO(3))
dx.
\end{equation}
The subscript $a$ in $\widetilde{R}_{a}^{(h)}$ is used
to remember that the rotation depends on the cylinder
$\widetilde{C}_{a,h}$. In particular, since $\Psi^{(h)}\big(\big(a, a +
\frac{L}{K_{h}}\big)\times D\big)\subset \widetilde{C}_{a,h}$, we
get
\begin{equation}\label{rigidity}
\int_{\Psi^{(h)}\big(\big(a, a +
\frac{L}{K_{h}}\big)\times
D\big)}\big|\,\nabla\tilde{v}^{(h)} -
\widetilde{R}_{a}^{(h)}\big|^{2} dx \leq c
\int_{\widetilde{C}_{a,h}}\mbox{dist}^{2}(\nabla\tilde{v}^{(h)},SO(3))
dx.
\end{equation}
Changing variables in the integral on the left-hand side, inequality (\ref{rigidity}) becomes
\begin{align*}
\int_{\big(a, a + \frac{L}{K_{h}}\big)\times
D}&\big|\,\nabla\tilde{v}^{(h)}\circ\Psi^{(h)} -
\widetilde{R}_{a}^{(h)}\big|^{2}\det
\big(\nabla\Psi^{(h)}\big) ds\,d\xi\,d\zeta \\
&\leq
c\,\int_{\widetilde{C}_{a,h}}\mbox{dist}^{2}\big(\nabla\tilde{v}^{(h)},SO(3)\big)
dx\\
&\leq
c\,\int_{\widetilde{C}_{a,h}}W\big(\big(\Psi^{(h)}\big)^{-1}
(x),\nabla\tilde{v}^{(h)}(x)\big) dx.
\end{align*}
Notice that $\det \big(\nabla\Psi^{(h)}\big) =
h^2 \det
\big(\nabla_{h}\Psi^{(h)}\big)$ and, since
$\det
\big(\nabla_{h}\Psi^{(h)}\big)\rightarrow 1$
uniformly,
\begin{equation}\label{rigidity2}
\int_{\big(a, a + \frac{L}{K_{h}}\big)\times
D}\big|\,\nabla\tilde{v}^{(h)}\circ\Psi^{(h)} -
\widetilde{R}_{a}^{(h)}\big|^{2} ds\,d\xi\,d\zeta \leq
\frac{c}{h^2}\,\int_{\widetilde{C}_{a,h}}W\big(\big(\Psi^{(h)}\big)^{-1}
(x),\nabla\tilde{v}^{(h)}(x)\big) dx.
\end{equation}
Now define the map $R^{(h)}: [0, L)\rightarrow SO(3)$
given by
\begin{equation*}
R^{(h)}(s):= \widetilde{R}_{a}^{(h)} \quad
\mbox{for}\, s\in \Big[a, a + \frac{L}{K_{h}}\Big),\,
a\in [0, L)\cap \frac{L}{K_{h}}\,\mathbb{N}.
\end{equation*}
Summing (\ref{rigidity2}) over $a\in [0, L)\cap
\frac{L}{K_{h}}\,\mathbb{N}$ leads to
\begin{equation*}
\int_{\Omega}\big|\,\nabla\tilde{
v}^{(h)}\circ\Psi^{(h)} -
R^{(h)}\big|^{2}ds\,d\xi\,d\zeta \leq
\frac{c}{h^2}\,\int_{\widetilde{\Omega}_{h}}W\big(\big(\Psi^{(h)}\big)^{-1}
(x),\nabla\tilde{v}^{(h)}(x)\big) dx
\end{equation*}
for a suitable constant independent of $h$. By (\ref{finite}) we
obtain
\begin{equation}\label{rigid2}
\int_{\Omega}\big|\,\nabla
\tilde{v}^{(h)}\circ\Psi^{(h)} -
R^{(h)}\big|^{2}ds\,d\xi\,d\zeta \leq c\,h^2.
\end{equation}
Now, applying iteratively estimate (\ref{rigidity2}) in neighbouring cubes, 
one can prove the following difference quotient estimate for $R^{(h)}$: 
for every $I'\subset\subset [0, L]$ and every $\delta \in \mathbb{R}$ with 
$\snorm{\,\delta} \leq \mbox{dist}(I',\{0, L\})$
\begin{equation}\label{increm}
\int_{I'}\big|\,R^{(h)}(s + \delta) - R^{(h)}(s)\,\big|^{2} ds
\leq c\,(|\,\delta\,| + h)^{2},
\end{equation}
with $c$ independent of $I'$ and $\delta$ (see \cite{MM03}, proof of Theorem 2.1).
Using the Fr\'echet-Kolmogorov criterion, we deduce that, for
every sequence $(h_{j})\rightarrow 0$, there exists a
subsequence of $R^{(h_{j})}$ which converges strongly in
$L^{2}(I';\mathbb{M}^{3\times 3})$ to some $\overline{R}\in L^{2}(I';\mathbb{M}^{3\times 3})$, with
$\overline{R}(s)\in SO(3)$ for a.e. $s\in I'$. From (\ref{Fconv}) and (\ref{rigid2})
it follows that $F = \overline{R}$ a.e.. Moreover (\ref{convdet}) and (\ref{bound2}) 
imply the convergence of the $L^{2}$ norm of $F^{(h)}$ to the $L^{2}$ norm of 
$\overline{R}$, hence 
\begin{equation*}
F^{(h)} \rightarrow \overline{R}  \quad \mbox{strongly
in} \,L^{2}(\Omega;\mathbb{M}^{3\times 3}).
\end{equation*}
This proves (\ref{teo3}), once the regularity of the function $\overline{R}$ is shown.
To this aim, divide both sides of the inequality (\ref{increm}) by
$(\snorm{\delta} + h)^{2}$ and let
$h\rightarrow 0$; then
\begin{equation}\label{regu}
\int_{I'}\frac{\snorm{\overline{R}(s + \delta) - \overline{R}(s)}^{2}}{\snorm{\delta}^{2}}\,ds \leq c
\end{equation}
and so $\overline{R} \in W^{1,2}(I';\mathbb{M}^{3\times 3})$. But
this holds for every $I'\subset\subset [0, L]$ with a constant
independent of the subset $I'$, hence $\overline{R} \in
W^{1,2}((0, L);\mathbb{M}^{3\times 3})$.
\newline
Now notice that
\begin{equation}\label{successione}
\nabla_{h}\big(\tilde{v}^{(h)}\circ\Psi^{(h)}\big)
= \big(\nabla\tilde{v}^{(h)}\circ\Psi^{(h)}\big)
\nabla_{h}\Psi^{(h)}=
F^{(h)}\nabla_{h}\Psi^{(h)};
\end{equation}
by (\ref{convdet}) and (\ref{teo3}) we deduce that
\begin{equation}\label{3.16bis}
\nabla_{h}\big(\tilde{v}^{(h)}\circ\Psi^{(h)}\big)\longrightarrow \overline{R}\,R_{0} \quad
\mbox{strongly in}\,\, L^{2}(\Omega;\mathbb{M}^{3\times 3}).
\end{equation}
In particular, we have
\begin{equation}\label{3.16ter}
\nabla\big(\tilde{v}^{(h)}\circ\Psi^{(h)}\big)\longrightarrow 
\big(\overline{R}\,R_{0}e_1\big)\otimes e_1 \quad
\mbox{strongly in}\,\, L^{2}(\Omega;\mathbb{M}^{3\times 3}).
\end{equation}
By Poincar\'e inequality there exist some constants $c^{(h)}\in \mathbb{R}^3$ and a function $v$ in $W^{1,2}(\Omega;\mathbb{R}^{3})$ such that (\ref{teo1}) is satisfied.
Moreover (\ref{3.16ter}) entails that the function $v$ depends only on the variable $s$ in $[0, L]$ and satisfies 
$v' = \overline{R}\,R_{0}e_1$. Setting $d_{k}:= \overline{R}\,R_{0}e_k$ for $k = 2,3$, we have that $(v, d_2, d_3)\in \mathcal{A}$ and (\ref{teo2}), (\ref{teo25}) are satisfied by (\ref{3.16bis}).
\end{proof}

\subsection{Bound from below}
Let $Q_{3}: \Omega\times\mathbb{M}^{3\times 3}\longrightarrow [0, +\infty)$ be twice the 
quadratic form of linearized elasticity; i.e.,
$$Q_{3}(z,G) := \frac{\partial^2 W}{\partial F^2}(z,Id)[G,G]$$
for a.e. $z\in \Omega$ and every $G\in \mathbb{M}^{3\times 3}$.
We introduce the quadratic form $Q_{2}: [0, L]\times\mathbb{M}^{3\times 3}_{\textrm{skew}}\rightarrow [0, +\infty)$
defined by
\begin{equation}\label{Q2}
Q_{2}(s,P):= \hspace{-0.4cm}\inf_{\begin{array}{c}
\vspace{-.15cm}
\scriptstyle\hat{\alpha}\in W^{1,2}(D;\mathbb{R}^{3})\\
\scriptstyle\hat{g}\in \mathbb{R}^{3}
\end{array}}\Bigg\{\int_{D}Q_{3}\bigg(s,\xi,\zeta,R_{0}(s)\Bigg(P\,
\Bigg(\begin{array}{c}
0\\
\xi\\
\zeta
\end{array}\Bigg) + \hat{g}\,\bigg|\,\partial_{\xi}\hat{\alpha}\,\bigg|\,\partial_{\zeta}\hat{\alpha}\Bigg)
R_{0}^{T}(s)\bigg)d\xi\,d\zeta\Bigg\}.
\end{equation}
\begin{rem}\label{remark}
It is easy to check that the minimum in (\ref{Q2}) is attained; 
moreover the minimizers depend linearly on $P$, hence $Q_2$ is a quadratic form of $P$. Notice also that if $P\in L^2 ((0,L);\mathbb{M}^{3\times 3})$, then $\hat{\alpha}\in L^2 (\Omega;\mathbb{R}^{3})$ with $\partial_{\xi}\hat{\alpha}, \partial_{\zeta}\hat{\alpha} \in L^2 (\Omega;\mathbb{R}^{3})$, and $\hat{g} \in L^2 ((0,L);\mathbb{R}^{3})$ (see \cite[Remarks 4.1 - 4.3]{MGMM04}).
\end{rem}
In the following theorem we prove a lower bound for the energies 
$\tilde{I}^{(h)}/h^{2}$ in terms of the functional
\begin{equation}\label{funI}
I(v,d_2,d_3):=
\left\{
\vspace{.7cm}
\begin{array}{ll}
\displaystyle\frac{1}{2}\int_{0}^{L} Q_{2} \big(s,\big(R^{T}(s)R'(s) -
R_{0}^{T}(s)R_{0}'(s)\big)\big)ds  & \mbox{if } \, (v,d_2,d_3)\in \mathcal{A},\\
\displaystyle + \infty & \mbox{otherwise }, 
\end{array}
\right.
\end{equation}
where $R\in W^{1,2}((0, L);\mathbb{M}^{3\times 3})$ denotes the matrix 
$R:= (v'\,|\,d_2\,|\,d_3)$ and $\mathcal{A}$ is the class defined in~(\ref{defA}).
\begin{thm}\label{scithm}
Let $v\in W^{1,2}(\Omega;\mathbb{R}^{3})$ and let $d_2, d_3 \in L^{2}(\Omega;\mathbb{R}^{3})$. Then, 
for every positive sequence $(h_j)$ converging to
zero and every sequence $\big(\tilde{v}^{(h_j)}\big)\subset
W^{1,2}(\widetilde{\Omega}_{h_j};\mathbb{R}^{3})$ such
that
\begin{equation}\label{teo5}
\tilde{v}^{(h_j)}\circ \Psi^{(h_j)} \rightarrow v \quad \mbox{strongly in}\,\,
W^{1,2}(\Omega;\mathbb{R}^{3}),
\end{equation}
\begin{equation}\label{teo65}
\frac{1}{h_j}\,\partial_\xi\big(\tilde{v}^{(h_j)}\circ
\Psi^{(h_j)}\big)\,\rightarrow d_2\quad
\mbox{strongly in}\,\, L^{2}(\Omega;\mathbb{R}^{3}),
\end{equation}
\begin{equation}\label{teo6}
\frac{1}{h_j}\,\partial_\zeta\big(\tilde{v}^{(h_j)}\circ
\Psi^{(h_j)}\big)\,\rightarrow d_3\quad
\mbox{strongly in}\,\, L^{2}(\Omega;\mathbb{R}^{3}),
\end{equation}
it turns out that
\begin{equation}\label{cl}
I(v,d_2,d_3)\leq\liminf_{j \rightarrow \infty}
\frac{1}{h_{j}^{4}} \int_{\widetilde{\Omega}_{h_j}}
W\big(\big(\Psi^{(h_j)}\big)^{-1} (x), \nabla
\tilde{v}^{(h_j)}(x)\big) dx.
\end{equation}
\end{thm}
\begin{proof}
In the following, for notational
brevity, we will write simply $h$ instead of $h_j$.
Let $\big(\tilde{v}^{(h)}\big)$ be a sequence satisfying (\ref{teo5}), (\ref{teo65}) and (\ref{teo6}). 
We can assume that 
$$\liminf_{j \rightarrow \infty}
\frac{1}{h_{j}^{4}} \int_{\widetilde{\Omega}_{h_j}}
W\big(\big(\Psi^{(h_j)}\big)^{-1} (x), \nabla
\tilde{v}^{(h_j)}(x)\big) dx \leq C < + \infty,$$
otherwise (\ref{cl}) is trivial. Therefore, up to subsequences, (\ref{finite}) is satisfied. By Theorem 
\ref{compactness} we deduce that $(v, d_2, d_3)\in \mathcal{A}$,
\begin{equation}\label{rel1}
F^{(h)}:= \nabla \tilde{v}^{(h)}\circ \Psi^{(h)} \longrightarrow \overline{R} \quad \mbox{strongly in } \,L^{2}(\Omega;\mathbb{M}^{3\times 3}) 
\end{equation}
with $\overline{R} \in W^{1,2}((0, L);\mathbb{M}^{3\times 3})$, $\overline{R}\in SO(3)$ a.e., and
\begin{equation}\label{rel2}
R:= (v'\,|\,d_2\,|\,d_3) = \overline{R}\,R_{0}.
\end{equation}
Moreover, as in the proof of Theorem \ref{compactness}, we can construct a piecewise constant approximation
$R^{(h)}: [0, L]\rightarrow SO(3)$ such that
\begin{equation}\label{rot}
\int_{\Omega} \big|F^{(h)} - R^{(h)}\big|^{2} ds\,d\xi\,d\zeta \leq c\,h^{2}
\end{equation}
and $R^{(h)} \rightarrow \overline{R}$ strongly in $L^{2}(I';\mathbb{M}^{3})$ for every 
$I'\subset\subset [0, L]$.
Define the functions $G^{(h)}: \Omega\rightarrow \mathbb{M}^{3\times 3}$\, as
\begin{equation}
G^{(h)}:= \frac{1}{h}\Big((R^{(h)})^{T} F^{(h)} - Id\Big) =
\frac{1}{h}\Big((R^{(h)})^{T} \nabla_{h}v^{(h)}\big(\nabla_{h}\Psi^{(h)}\big)^{-1} - Id\Big).
\end{equation}
By (\ref{rot}) they are bounded in $L^{2}(\Omega;\mathbb{M}^{3\times 3})$, so there exists $G\in L^{2}(\Omega;\mathbb{M}^{3\times 3})$ such that $G^{(h)}\rightharpoonup G$ weakly in $L^{2}(\Omega;\mathbb{M}^{3\times 3})$. We claim that
\begin{equation}\label{liminf}
\liminf_{h\rightarrow 0}\frac{1}{h^4}\int_{\widetilde{\Omega}_{h}} W\Big(\big(\Psi^{(h)}\big)^{-1} (x),\nabla \tilde{v}^{(h)}(x)\Big) dx \geq \frac{1}{2} \int_{\Omega} Q_{3}(s,\xi,\zeta,G) ds\,d\xi\,d\zeta.
\end{equation}
Performing the change of variables $\Psi^{(h)}$, we have
\begin{align}\label{newvar}
\frac{1}{h^4}\int_{\widetilde{\Omega}_{h}} W\Big(\big(\Psi^{(h)}\big)^{-1} (x),\nabla \tilde{v}^{(h)}(x)\Big) dx & = \frac{1}{h^2}\int_{\Omega} W\big(s,\xi,\zeta, F^{(h)}\big)\det \big(\nabla_{h}\Psi^{(h)}\big) ds\,d\xi\,d\zeta \nonumber\\
& = \frac{1}{h^2}\int_{\Omega} W\big(s,\xi,\zeta, \big(R^{(h)}\big)^T F^{(h)}\big)\det \big(\nabla_{h}\Psi^{(h)}\big) ds\,d\xi\,d\zeta
\end{align}
where the last equality follows from the frame indifference of $W$.
Define the family of functions
$$
\chi^{(h)}(s,\xi,\zeta):=
\left\{
\begin{array}{ll}
\vspace{.1cm}
\displaystyle 1 & \mbox{in } \,\, \Omega\cap\big\{(s,\xi,\zeta): \snorm{G^{(h)}(s,\xi,\zeta)} \leq h^{-\frac{1}{2}}\big\},\\
\displaystyle 0 & \mbox{otherwise}.
\end{array}
\right.
$$
From the boundedness of $G^{(h)}$ in $L^2(\Omega;\mathbb{M}^{3\times 3})$ we get that \,\,
$\chi^{(h)}\rightarrow 1$ \,\,boundedly in measure, so that
\begin{equation}\label{convchi}
\chi^{(h)}G^{(h)} \rightharpoonup G \quad \mbox{weakly in}\, L^{2}(\Omega;\mathbb{M}^{3\times 3}).
\end{equation}
By expanding $W$ around the identity, we obtain that for every $(s,\xi,\zeta) \in \Omega$ and 
$A\in \mathbb{M}^{3\times 3}$
\begin{equation*}
W\big(s,\xi,\zeta, Id + A) = \frac{1}{2}\,\frac{\partial^{2}W}{\partial F^{2}}\,(s,\xi,\zeta, Id + t\,A)[A,A]
\end{equation*}
where \,$0<t<1$\, depends on the point $(s,\xi,\zeta)$ and on $A$. By (\ref{newvar}) and by the definition of $G^{(h)}$ 
\begin{align}
\frac{1}{h^2}\,\tilde{I}^{(h)}\big(\tilde{v}^{(h)}\big) &=
\frac{1}{h^2}\int_{\Omega} W\big(s,\xi,\zeta, Id + h\,G^{(h)} \big)\det \big(\nabla_{h}\Psi^{(h)}\big) ds\,d\xi\,d\zeta \nonumber\\
&\geq \frac{1}{h^2}\int_{\Omega} \chi^{(h)} W\big(s,\xi,\zeta,Id + h\,G^{(h)}\big)\det \big(\nabla_{h}\Psi^{(h)}\big) ds\,d\xi\,d\zeta \nonumber\\
&= \frac{1}{2}\int_{\Omega} \chi^{(h)}\left(\frac{\partial^{2}W}{\partial F^{2}}\,\big(s,\xi,\zeta, Id + h\,t(h)\,G^{(h)}\big)\big[G^{(h)},G^{(h)}\big]\right)\det \big(\nabla_{h}\Psi^{(h)}\big) ds\,d\xi\,d\zeta\label{3.22ter}
\end{align}
where $0 < t(h) < 1$\, depends on $(s,\xi,\zeta)$ and on $G^{(h)}$.
The last integral in the previous formula can be written as
\begin{align}
&\frac{1}{2}\int_{\Omega} \chi^{(h)}\left(\frac{\partial^{2}W}{\partial F^{2}}\,\big(s,\xi,\zeta, Id + h\,t(h)\,G^{(h)}\big)\big[G^{(h)},G^{(h)}\big]\right)\det \big(\nabla_{h}\Psi^{(h)}\big) ds\,d\xi\,d\zeta \nonumber\\
& = \frac{1}{2}\int_{\Omega}\Big( \chi^{(h)}\bigg(\frac{\partial^{2}W}{\partial F^{2}}\,\big(s,\xi,\zeta, Id + h\,t(h)\,G^{(h)}\big)\big[G^{(h)},G^{(h)}\big] - Q_{3}\big(s,\xi,\zeta, G^{(h)}\big)\Big)\Big)\det \big(\nabla_{h}\Psi^{(h)}\big) ds\,d\xi\,d\zeta \nonumber\\
&\hspace{1cm}+ \frac{1}{2}\int_{\Omega}Q_{3}\big(s,\xi,\zeta,\chi^{(h)}\,G^{(h)}\big)\det \big(\nabla_{h}\Psi^{(h)}\big) ds\,d\xi\,d\zeta. \label{3.22quater}
\end{align}
By Scorza-Dragoni theorem there exists a compact set $K\subset \Omega$ such that the function \,$\partial^{2}W/\partial F^{2}$ \, restricted to $K\times \overline{B_{\delta}(Id)}$ is continuous, hence uniformly continuous.
Since $h\,t(h)\,\chi^{(h)}\,G^{(h)}$ is uniformly small for $h$ small enough, for every $\varepsilon > 0$ we have
\begin{align*}
\frac{1}{2}\int_{\Omega}&\chi^{(h)}\Bigg(\frac{\partial^{2}W}{\partial F^{2}}\,\big(s,\xi,\zeta, Id + h\,t(h)\,G^{(h)}\big)\big[G^{(h)},G^{(h)}\big] - Q_{3}\big(s,\xi,\zeta,G^{(h)}\big)\Bigg)\det \big(\nabla_{h}\Psi^{(h)}\big) ds\,d\xi\,d\zeta \\
&\geq - \frac{\varepsilon}{2}\int_{K}\chi^{(h)}\big|\,G^{(h)}\big|^{2}\det \big(\nabla_{h}\Psi^{(h)}\big) ds\,d\xi\,d\zeta
\geq -\,C\,\varepsilon
\end{align*}
for $h$ small enough. As for the second integral in (\ref{3.22quater}), 
by (\ref{convdet}) and (\ref{newvar}) we get 
\begin{equation}\label{3.22quinq}
\liminf_{h\rightarrow 0}\frac{1}{2}\int_{\Omega}Q_{3}\big(s,\xi,\zeta,\chi^{(h)}\,G^{(h)}\big)\det \big(\nabla_{h}\Psi^{(h)}\big) ds\,d\xi\,d\zeta \geq \frac{1}{2}\int_{\Omega}Q_{3}\big(s,\xi,\zeta,G\big)ds\,d\xi\,d\zeta
\end{equation}
since $Q_{3}$ is a nonnegative quadratic form.
Combining (\ref{3.22ter}), (\ref{3.22quater}) and (\ref{3.22quinq}) we have
\begin{equation*}
\liminf_{h\rightarrow 0}\frac{1}{h^2}\,\tilde{I}^{(h)}(\tilde{v}^{(h)})
\geq \frac{1}{2}\int_{\Omega}Q_{3}(s,\xi,\zeta,G) ds\,d\xi\,d\zeta \,-\,C\,\varepsilon
\end{equation*}
and, since $\varepsilon$ is arbitrary, (\ref{liminf}) is proved. It remains to identify $G$. 

Fix $(\xi_0,\zeta_0)\in D$; let $\delta_0 = \delta_{0}(\xi_0,\zeta_0) > 0$ be such that 
$B_{2\,\delta_0}(\xi_0,\zeta_0)\subset D$ and let $U_0:= (0, L)\times B_{\delta_0}(\xi_0,\zeta_0)$.
Fix $t\in \mathbb{R}-\{0\}$, $|\,t\,|\,<\,\delta_0$. 
For every $(s,\xi,\zeta)\in U_0$ we can define 
the difference quotients of the functions $G^{(h)}$ with respect to 
the variables $\xi$ and $\zeta$ along the direction $\tau$, given by
$$
\left\{
\begin{array}{ll}
H^{(h)}_{t}(s,\xi,\zeta):=& \dfrac{1}{t}\,\Big(G^{(h)}(s,\xi + t,\zeta) - G^{(h)}(s,\xi,\zeta)\Big)\,\tau(s),\\
\mbox{ }\\
K^{(h)}_{t}(s,\xi,\zeta):=& \dfrac{1}{t}\,\Big(G^{(h)}(s,\xi,\zeta + t) - G^{(h)}(s,\xi,\zeta)\Big)\,\tau(s),  
\end{array}
\right.
$$
and the corresponding difference quotients of the limit function $G$
$$
\left\{
\begin{array}{ll}
H_{t}(s,\xi,\zeta):=& \dfrac{1}{t}\,\Big(G(s,\xi + t,\zeta) - G(s,\xi,\zeta)\Big)\,\tau(s),\\
\mbox{ }\\
K_{t}(s,\xi,\zeta):=& \dfrac{1}{t}\,\Big(G(s,\xi,\zeta + t) - G(s,\xi,\zeta)\Big)\,\tau(s).  
\end{array}
\right.
$$
Since $G^{(h)}\rightharpoonup G$ in $L^{2}(\Omega;\mathbb{M}^{3\times 3})$ and 
$R^{(h)} \longrightarrow \overline{R}$ boundedly in measure, we have
\begin{align}\label{convH}
H^{(h)}_{t}&\rightharpoonup H_{t}
\quad \mbox{weakly in} \, L^{2}(U_{0};\mathbb{R}^{3}) \,\, \mbox{and} \nonumber\\
R^{(h)}\,H^{(h)}_{t} &\rightharpoonup \overline{R}\,H_{t}
\quad \mbox{weakly in} \, L^{2}(U_{0};\mathbb{R}^{3}).  
\end{align}
In terms of $F^{(h)}$ the left-hand side of (\ref{convH}) reads as
\begin{equation}\label{riscritta}
R^{(h)}(s) H^{(h)}_{t}(s,\xi,\zeta) = 
\frac{1}{h\,t}\,\Big(F^{(h)}(s,\xi + t,\zeta) - F^{(h)}(s,\xi,\zeta)\Big)\,\tau(s).
\end{equation}
Now recall that, if we set $ v^{(h)}:= \tilde{v}^{(h)}\circ \Psi^{(h)}$, we have
\begin{equation}\label{VvsF}
\nabla v^{(h)} = F^{(h)} \, \nabla \Psi^{(h)}; 
\end{equation}
in particular, taking the first column of the two matrices, we obtain
\begin{equation*}
F^{(h)}(s,\xi,\zeta)\,\tau(s) = \partial_{s}v^{(h)}(s,\xi,\zeta) - h\,F^{(h)}(s,\xi,\zeta)\,(\xi\,\nu'_{2}(s) + \zeta\,\nu'_{3}(s)).
\end{equation*}
By the last equality and (\ref{riscritta}) we get
\begin{align}\label{lunga}
R^{(h)}(s) H^{(h)}_{t}(s,\xi,\zeta) 
=& \,\frac{1}{h\,t}\,\Big(\partial_{s} v^{(h)}(s,\xi + t,\zeta) - \partial_{s} v^{(h)}(s,\xi,\zeta)\Big)\nonumber\\ 
-&\,\frac{1}{t}\Big((\xi + t)\,F^{(h)}(s,\xi + t,\zeta) - \xi\,F^{(h)}(s,\xi,\zeta)\Big)\,\nu_{2}'(s)\nonumber\\
-&\,\frac{1}{t}\Big(\zeta\,F^{(h)}(s,\xi + t,\zeta) - \zeta\,F^{(h)}(s,\xi,\zeta)\Big)\,\nu_{3}'(s).
\end{align}
For the first term we have
\begin{align*}
\frac{1}{h\,t}\,\partial_{s}\Big(v^{(h)}(s,\xi + t,\zeta) - v^{(h)}(s,\xi,\zeta)\Big)
& = \frac{1}{h\,t}\,\partial_{s}\bigg(\int_{\xi}^{\xi + t}\partial_{\xi}v^{(h)}(s,\vartheta,\zeta)\, d\vartheta\bigg)\\
& = \partial_{s}\bigg(\frac{1}{t}\int_{0}^{t}\frac{1}{h}\,\partial_{\xi}v^{(h)}(s,\xi +\vartheta,\zeta)\,d\vartheta\bigg),
\end{align*}
so by (\ref{teo6}) and (\ref{rel2})
\begin{equation}\label{RHS1}
\frac{1}{h\,t}\,\partial_{s}\Big(v^{(h)}(s,\xi + t,\zeta) - v^{(h)}(s,\xi,\zeta)\Big) \rightharpoonup  d'_2 (s) = \partial_{s} (\overline{R}(s)\,\nu_2(s)) \quad\mbox{weakly in}\, W^{-1,2}(U_{0};\mathbb{R}^{3}).
\end{equation}
By (\ref{rel1}) the second term in (\ref{lunga}) converges to
\begin{equation}\label{RHS2}
\frac{1}{t}\,\Big((\xi + t)\,\overline{R}(s) - \xi\,\overline{R}(s)\Big)\,\nu_{2}'(s) = \overline{R}(s)\,\nu_2'(s) \quad\mbox{strongly in}\, L^{2}(U_0;\mathbb{R}^{3})
\end{equation}
and the last term to
\begin{equation}\label{RHS3}
\frac{1}{t}\,\Big(\zeta\,\overline{R}(s) - \zeta\,\overline{R}(s)\Big)\,\nu_{3}'(s) = 0 \quad\mbox{strongly in}\, L^{2}(U_0;\mathbb{R}^{3}).
\end{equation}
Putting together (\ref{RHS1}), (\ref{RHS2}), (\ref{RHS3}) and (\ref{convH})
\begin{equation*}
\overline{R}(s)\,H_{t}(s,\xi,\zeta) = \partial_{s} (\overline{R}(s)\,\nu_2(s)) - \overline{R}(s)\,\nu_2'(s) \, \,\mbox{a.e. in } \, U_0
\end{equation*}
and so
\begin{equation}\label{acca}
H_{t}(s,\xi,\zeta) = (\overline{R}(s))^{T}\,\overline{R}'(s)\,\nu_2(s) \,\,\mbox{a.e. in } \, U_0.
\end{equation}
Repeating the same argument for $K^{(h)}_{t}$ we get
\begin{equation}\label{kappa}
K_{t}(s,\xi,\zeta) = (\overline{R}(s))^{T}\,\overline{R}'(s)\,\nu_3(s) \,\,\mbox{a.e. in } \, U_0.
\end{equation}
From the last two equalities we deduce that the functions $H_{t}$ and $K_{t}$ depend only on the variable $s$.
Moreover, letting $t$ go to $0$ both in (\ref{acca}) and in (\ref{kappa}), we get that the gradient of $G\,\tau$ w.r.to the variables $(\xi,\zeta)$ depends only on $s$, i.e.,
\begin{equation}\label{gradgrad}
\nabla_{(\xi,\zeta)}\big( G(s,\xi,\zeta)\,\tau(s)\big) = (\overline{R}(s))^{T}\,\overline{R}'(s)\,(\nu_2(s)\,|\,\nu_3(s)) \,\mbox{a.e. in } \, U_0.
\end{equation}
Being this equality valid in $U_0 = (0, L)\times B_{\delta_0}(\xi_0,\zeta_0)$, for an arbitrary $(\xi_0,\zeta_0)\in D$, we can conclude that it holds a.e. in the whole $\Omega$. Since $D$ is connected, we obtain that for a.e. $(s,\xi,\zeta) \in \Omega$
\begin{equation*}
G(s,\xi,\zeta)\,\tau(s) = (\overline{R}(s))^{T}\,\overline{R}'(s)\,(\xi\,\nu_2(s) + \zeta\,\nu_3(s)) + g(s)
\end{equation*}
with $g: [0, L]\rightarrow \mathbb{R}^{3}$. Remark that from the previous formula 
$g\in L^{2}((0, L);\mathbb{R}^{3})$.
\newline
It remains to identify the components $G(s,\xi,\zeta)\,\nu_{2}(s)$ and $G(s,\xi,\zeta)\,\nu_{3}(s)$.
By (\ref{VvsF}) we have
\begin{eqnarray*}
G^{(h)}(s,\xi,\zeta)\,\nu_{2}(s) &=& \frac{1}{h}\,\Big((R^{(h)}(s))^{T} F^{(h)}(s,\xi,\zeta)\,\nu_{2}(s) - \nu_{2}(s)\Big)\\
&=& \frac{1}{h}\,\Big(h^{-1}(R^{(h)}(s))^{T}\partial_{\xi}v^{(h)}(s,\xi,\zeta) - \nu_{2}(s)\Big)
\end{eqnarray*}
and
\begin{eqnarray*}
G^{(h)}(s,\xi,\zeta)\,\nu_{3}(s) &=& \frac{1}{h}\,\Big((R^{(h)}(s))^{T} F^{(h)}(s,\xi,\zeta)\,\nu_{3}(s) - \nu_{3}(s)\Big)\\
&=& \frac{1}{h}\,\Big(h^{-1}(R^{(h)}(s))^{T}\partial_{\zeta}v^{(h)}(s,\xi,\zeta) - \nu_{3}(s)\Big),
\end{eqnarray*}
so, if we define
\begin{equation*}
\alpha^{(h)}(s,\xi,\zeta):= \frac{1}{h}\,\Big(h^{-1}(R^{(h)})^{T} v^{(h)}(s,\xi,\zeta) - \xi\,\nu_{2}(s) - \zeta\,\nu_{3}(s)\Big)
\end{equation*}
it turns out that
\begin{equation}\label{alfa}
\partial_{\xi}\alpha^{(h)}(s,\xi,\zeta) = G^{(h)}(s,\xi,\zeta)\,\nu_{2}(s)\quad\mbox{and}\quad
\partial_{\zeta}\alpha^{(h)}(s,\xi,\zeta) = G^{(h)}(s,\xi,\zeta)\,\nu_{3}(s).
\end{equation}
Applying the Poincar\'e inequality  to the functions $\alpha^{(h)}$ for fixed $s$ we obtain that 
for a.e. $s\in [0, L]$
\begin{equation*}
\int_{D}\big|\,\alpha^{(h)}(s,\xi,\zeta) - \alpha_{0}^{(h)}(s)\,\big|^{2}\,d\xi\,d\zeta \leq c\int_{D}\left(\big|\,\partial_{\xi}\alpha^{(h)}(s,\xi,\zeta)\big|^{2} + \big|\,\partial_{\zeta}\alpha^{(h)}(s,\xi,\zeta)\big|^{2}\right)\,d\xi\,d\zeta,
\end{equation*}
where $\alpha_{0}^{(h)}(s):= \int_{D}\alpha^{(h)}(s,\xi,\zeta)\,d\xi\,d\zeta$. Integrating over $[0, L]$, we have
\begin{equation*}
\big|\big|\alpha^{(h)} - \alpha_{0}^{(h)}\big|\big|^{2}_{L^{2}(\Omega)} \leq c\left(\big|\big|\partial_{\xi}\alpha^{(h)}\big|\big|^{2}_{L^{2}(\Omega)} +\big|\big|\partial_{\zeta}\alpha^{(h)}\big|\big|^{2}_{L^{2}(\Omega)}\right).
\end{equation*}
Since the right-hand side is bounded by (\ref{alfa}), there exists 
a function $\alpha\in L^{2}(\Omega; \mathbb{R}^{3})$ such that,
up to subsequences,
$$\alpha^{(h)} - \alpha_{0}^{(h)}\rightharpoonup \alpha \quad \mbox{weakly in} \,L^{2}(\Omega; \mathbb{R}^{3}).$$
Moreover, from (\ref{alfa}) we conclude that
\begin{equation}\label{alfa1}
\partial_{\xi}\alpha(s,\xi,\zeta) = G(s,\xi,\zeta)\,\nu_{2}(s)\quad\mbox{and}\quad
\partial_{\zeta}\alpha(s,\xi,\zeta) = G(s,\xi,\zeta)\,\nu_{3}(s),
\end{equation}
therefore $\partial_{\xi}\alpha,\partial_{\zeta}\alpha \in L^{2}(\Omega; \mathbb{R}^{3})$.
Now, define the functions $\hat{\alpha}(s,\xi,\zeta):= R_{0}^{T}(s)\,\alpha(s,\xi,\zeta)$ and  $\hat{g}(s):= R_{0}^{T}(s)\,g(s)$.
Thanks to these definitions and to (\ref{rel2}), $G$ can be written as
\begin{align}\label{tildeG}
G =& \Bigg(\Bigl(R\,R_{0}^{T}\Bigr)^{T} \Bigl(R\,R_{0}^{T}\Bigr)' R_{0}\Bigg(\begin{array}{c}
0\\
\xi\\
\zeta
\end{array}\Bigg) + g\,\bigg|\,\partial_{\xi}\alpha\,
\bigg|\,\partial_{\zeta}\alpha\Bigg)\,R^{T}_{0}\nonumber\\
=&\, R_{0}\Bigg(\Bigl( R^{T}R' + (R_{0}^{T})'R_{0}\Bigr)\Bigg(\begin{array}{c}
0\\
\xi\\
\zeta
\end{array}\Bigg) + \hat{g} \,\bigg|\,\partial_{\xi}\hat{\alpha}\,
\bigg|\,\partial_{\zeta}\hat{\alpha}\Bigg)\,R_{0}^{T} \nonumber\\
=&\, R_{0}\Bigg(\Bigl( R^{T}R' - R_{0}^{T}R_{0}'\Bigr)\Bigg(\begin{array}{c}
0\\
\xi\\
\zeta
\end{array}\Bigg) + \hat{g}\,\bigg|\,\partial_{\xi}\hat{\alpha}\,
\bigg|\,\partial_{\zeta}\hat{\alpha}\Bigg)R_{0}^{T},
\end{align}
where the last equality follows from the identity \, $\big(R_{0}^{T}\big)'R_{0} + R_{0}^{T}R_{0}' = 0$.
Combining (\ref{liminf}) and (\ref{tildeG}), we obtain
\begin{equation*}
\liminf_{h\rightarrow 0}\frac{1}{h^2}\,\tilde{I}^{(h)}(\tilde{v}^{(h)})
\geq\frac{1}{2}\int_{\Omega} Q_{3}\bigg(s,\xi,\zeta,R_{0}(s)\Bigg(P(s)\Bigg(\begin{array}{c}
0\\
\xi\\
\zeta
\end{array}\Bigg) + \hat{g}\,\bigg|\,\partial_{\xi}\hat{\alpha}\,
\bigg|\,\partial_{\zeta}\hat{\alpha}\Bigg)R_{0}^{T}(s)\bigg)ds\,d\xi\,d\zeta,
\end{equation*}
with $P(s):= R^{T}(s)R'(s) - R_{0}^{T}(s)R_{0}'(s)$.
By the definition of the quadratic form $Q_2$ in (\ref{Q2}) we clearly have $\int_{D} Q_{3}(s,\xi,\zeta,G)d\xi\,d\zeta \geq Q_{2}(s,P(s))$, and so
\begin{equation*}
\liminf_{h\rightarrow 0} \frac{1}{h^4}
\int_{\widetilde{\Omega}_{h}} W\big(\big(\Psi^{(h)}\big)^{-1} (x),
\nabla \tilde{v}^{(h)}(x)\big) dx \geq  \frac{1}{2}\int_{0}^{L} Q_{2}
\big(s,\big(R^{T}(s)R'(s) - R_{0}^{T}(s)R_{0}'(s)\big)\big)ds.
\end{equation*}
\end{proof}

\subsection{Bound from above}
In this subsection we show that the lower bound proved in Theorem~\ref{scithm} is optimal.
\begin{thm}\label{bfa}
For every sequence of positive $(h_j)$ converging to $0$ and for every $(v,d_2,d_3)\in \mathcal{A}$ there exists a sequence $\big(\tilde{v}^{(h_j)}\big) \subset W^{1,2}\big(\widetilde{\Omega}_{h_j}; \mathbb{R}^{3}\big)$ such that 
\begin{align}
\tilde{v}^{(h_j)}\circ\Psi^{(h_j)} &\rightarrow v \quad \mbox{strongly in}\,\,
W^{1,2}(\Omega; \mathbb{R}^{3}),\label{star}\\
\frac{1}{h_j}\,\partial_\xi \big(\tilde{v}^{(h_j)}\circ\Psi^{(h_j)} \big)\,&\rightarrow d_2\quad
\mbox{strongly in}\,\, L^{2}(\Omega; \mathbb{R}^{3}), \label{starr1}\\
\frac{1}{h_j}\,\partial_\zeta \big(\tilde{v}^{(h_j)}\circ\Psi^{(h_j)}\big)\,&\rightarrow d_3\quad
\mbox{strongly in}\,\, L^{2}(\Omega; \mathbb{R}^{3}), \label{star2}
\end{align}
and
\begin{equation}\label{starec}
I(v,d_2,d_3) = \lim_{j \rightarrow \infty}
\frac{1}{h_{j}^4}\int_{\widetilde{\Omega}_{h_j}} W\big(\big(\Psi^{(h_j)}\big)^{-1} (x),
\nabla \tilde{v}^{(h_j)}(x)\big) dx, 
\end{equation}
where the class $\mathcal{A}$ and the functional $I$ are defined in (\ref{defA}) and (\ref{funI}), respectively.
\end{thm}
\begin{proof}
Let $(v,d_{2},d_{3})\in \mathcal{A}$. Assume in addition that 
$v\in C^{2}([0,L];\mathbb{R}^{3})$\,and\,  $d_{2},d_{3}\in C^{1}([0,L];\mathbb{R}^{3})$. 
Consider the functions $v^{(h)}: \Omega\rightarrow \mathbb{R}^{3}$ defined by
$$v^{(h)}(s,\xi,\zeta):= v(s) + h\,\xi\,d_{2}(s) + h\,\zeta\,d_{3}(s) + h\,q(s) +  h^{2}\,\beta(s,\xi,\zeta),$$
with $q\in C^{1}([0,L];\mathbb{R}^{3})$ and $\beta\in C^{1}(\overline{\Omega};\mathbb{R}^{3})$. 
We define $\tilde{v}^{(h)} := v^{(h)}\circ\big(\Psi^{(h)}\big)^{-1}$;
these functions clearly satisfy (\ref{star}). Moreover, since
\begin{equation}\label{starec2}
\nabla_{h} \big(\tilde{v}^{(h)}\circ\Psi^{(h)}\big) = 
\nabla_{h} v^{(h)} = (v'\,|\,d_{2}\,|\,d_{3}) + h\,\big(\xi\,d'_2 + \zeta\,d'_3 + q'\,|\,\partial_\xi\beta\,|\,\partial_\zeta\beta\big) + h^2\partial_s \beta\otimes e_1,
\end{equation}
also (\ref{starr1}) and (\ref{star2}) follow easily.
In order to prove (\ref{starec}), we first observe that, performing the change of variables $(s,\xi,\zeta) = \big(\Psi^{(h)}\big)^{-1}(x)$, we obtain 
\begin{align}
\frac{1}{h^2}\,\tilde{I}^{(h)}\big(\tilde{v}^{(h)}\big) =& 
\frac{1}{h^2}\,\int_{\Omega} W\big(s,\xi,\zeta,\nabla\tilde{v}^{(h)}\circ
\Psi^{(h)}\big)\det
\big(\nabla_{h}\Psi^{(h)}\big)
ds\,d\xi\,d\zeta \nonumber\\
=& \frac{1}{h^2}\,\int_{\Omega} W\big(s,\xi,\zeta,\nabla_h \big(\tilde{v}^{(h)}\circ
\Psi^{(h)}\big)\,\big(\nabla_h\Psi^{(h)}\big)^{-1}\big)\det
\big(\nabla_{h}\Psi^{(h)}\big)
ds\,d\xi\,d\zeta,\label{flower2}
\end{align}
where the last equality is justified observing that
\begin{equation*}
\nabla_h\big(\tilde{v}^{(h)}\circ
\Psi^{(h)}\big) = \big(\nabla\tilde{v}^{(h)}\circ
\Psi^{(h)}\big)\,\big(\nabla_h\Psi^{(h)}\big).
\end{equation*}
Then, by the definition of $\tilde{v}^{(h)}$,
\begin{equation}\label{flower}
\frac{1}{h^2}\,\tilde{I}^{(h)}\big(\tilde{v}^{(h)}\big) = \frac{1}{h^2}\,\int_{\Omega} W\big(s,\xi,\zeta,\big(\nabla_h v^{(h)}\big)\,\big(\nabla_h\Psi^{(h)}\big)^{-1}\big)\det
\big(\nabla_{h}\Psi^{(h)}\big)
ds\,d\xi\,d\zeta.
\end{equation}
Using (\ref{invA}) and (\ref{starec2}) we get
\begin{align*}
\nabla_{h} v^{(h)}\,\big(\nabla_{h}\Psi^{(h)}\big)^{-1} =&\,
R\, R_{0}^{T} + h\,(\xi\,d\,'_{2} + \zeta\,d\,'_{3} + q'\,|\,\partial_{\xi}\beta\,|\,\partial_{\zeta}\beta)\,R_{0}^{T} \\
-& \,h\, R\,R_{0}^{T}\big[(\xi\,\nu'_{2} + \zeta\,\nu'_{3})\otimes e_1\big]\,R_{0}^{T} + O(h^{2}),
\end{align*}
where $R = (v'|d_2|d_3)$ and $O(h^{2})/h^2$ is uniformly bounded.
Now consider the rotation $\overline{R}(s) = R(s)\, R_{0}^{T}(s)$. Then
\begin{equation*}
\overline{R}^{T}\nabla_{h} v^{(h)}\,\big(\nabla_{h}\Psi^{(h)}\big)^{-1} =\,
Id + h\,\overline{R}^{T}(\xi\,d\,'_{2} +\,\zeta\,d\,'_{3} + q'\,|\,\partial_{\xi}\beta\,|\,\partial_{\zeta}\beta)\,R_{0}^{T} - \,h \,\big[(\xi\,\nu'_{2} + \zeta\,\nu'_{3})\otimes e_1\big]\,R_{0}^{T} + O(h^{2}).
\end{equation*}
If we define the functions
\begin{equation*}
B^{(h)}(s,\xi,\zeta):= \,
\frac{1}{h}\bigg(\overline{R}^{T}\, \nabla_{h}v^{(h)}\,\big(\nabla_{h}\Psi^{(h)}\big)^{-1} - Id \bigg),
\end{equation*}
it turns out that
\begin{align}
B^{(h)} &=\, (R_{0}\, R^{T})(\xi\,d\,'_{2} +\zeta\,d\,'_{3} + q'\,|\,\partial_{\xi}\beta\,|\,\partial_{\zeta}\beta)R_{0}^{T}
- \big[(\xi\,\nu'_{2} + \zeta\,\nu'_{3})\otimes e_1\big]\,R_{0}^{T} + O(h)\nonumber\\
&=\,
R_{0}\, R^{T}\Bigg(R'\Bigg(\begin{array}{c}
0\\
\xi\\
\zeta
\end{array}\Bigg) + q' \,\bigg|\,
\partial_{\xi}\beta\,\bigg|\,\partial_{\zeta}\beta\Bigg)R_{0}^{T} - \Bigg[
\Bigg(R_{0}'\Bigg(
\begin{array}{c}
0\\
\xi\\
\zeta
\end{array}
\Bigg)\Bigg)\otimes e_1\Bigg]\,R_{0}^{T} + O(h)\nonumber\\
&=\,
R_{0}\Bigg(\Big(R^{T} R' - R_{0}^{T} R_{0}'\Big)\Biggl(\begin{array}{c}
0\\
\xi\\
\zeta
\end{array}\Biggr) + R^{T}q\,\bigg|\,R^{T}
\partial_{\xi}\beta\,\bigg|\,R^{T}\partial_{\zeta}\beta\Bigg)R_{0}^{T} + O(h)\nonumber\\
&=:\, G_{q,\beta} + O(h)\label{rombo}
\end{align}
where $O(h)/h$ is uniformly bounded.
By frame indifference and the definition of $B^{(h)}$, we have
\begin{align*}
\frac{1}{h^{2}}\,W\big(s,\xi,\zeta,\nabla_{h} v^{(h)}\big(\nabla_{h}\Psi^{(h)}\big)^{-1}) &= \, \frac{1}{h^{2}}\,W\big(s,\xi,\zeta,\overline{R}^{T} \nabla_{h} v^{(h)}\big(\nabla_{h}\Psi^{(h)}\big)^{-1}) \\
&= \,\frac{1}{h^{2}}\,W\big(s,\xi,\zeta,Id + h\,B^{(h)}\big).
\end{align*}
Using (\ref{rombo}) and the expansion of $W$ around the identity, we obtain
\begin{equation*}
\frac{1}{h^{2}}\,W\big(s,\xi,\zeta, \nabla_{h} v^{(h)}\big(\nabla_{h}\Psi^{(h)}\big)^{-1})\rightarrow \frac{1}{2}\,Q_{3}(s,\xi,\zeta,G_{q,\beta}) \quad\mbox{a.e.}.
\end{equation*}
Moreover, the assumption (ii) gives the uniform bound
\begin{equation*}
\frac{1}{h^{2}}\,W\big(s,\xi,\zeta, \nabla_{h} v^{(h)}\big(\nabla_{h}\Psi^{(h)}\big)^{-1}) \leq
\frac{1}{2}\,C_{1}\,|\,G_{q,\beta}\,|^{2} + C \in\, L^{1}(\Omega),
\end{equation*}
so, by the dominated convergence theorem and by (\ref{flower}) we conclude that
\begin{equation}\label{gammasup}
\lim_{h\rightarrow 0}\frac{1}{h^4}\int_{\widetilde{\Omega}_{h}} W\Big(\big(\Psi^{(h)}\big)^{-1} (x),\nabla \tilde{v}^{(h)}(x)\Big) dx
= \frac{1}{2}\int_{\Omega} Q_{3}(s,\xi,\zeta,G_{q,\beta})\,ds\,d\xi\,d\zeta .
\end{equation}
This holds for every $q\in C^{1}([0,L];\mathbb{R}^{3})$ and for every $\beta\in C^{1}(\overline{\Omega};\mathbb{R}^{3})$. 

Consider now the general case. Let $(v,d_{2},d_{3})\in\mathcal{A}$, and let $\hat{\alpha}(s,\cdot)\in W^{1,2}(D;\mathbb{R}^{3})$, $\hat{g}(s)$ be a solution to the minimum problem (\ref{Q2}) for $P = R^T R' - R_0^T R'_0$. 
By Remark \ref{remark}, $\hat{\alpha}\in L^{2}(\Omega;\mathbb{R}^{3})$ with $\partial_{\xi}\hat{\alpha}, \partial_{\zeta}\hat{\alpha}\in L^{2}(\Omega;\mathbb{R}^{3})$ and $\hat{g}\in L^{2}((0,L);\mathbb{R}^{3})$. 
In order to conclude the proof it is enough to construct a sequence of smooth deformations converging to $(v,d_{2},d_{3})$, on which the energy $\tilde{I}^{(h)}/h^2$ converges to the right-hand side of (\ref{gammasup}) with $q$ and $\beta$ replaced by $R^T \hat{g}$ and $R^T \hat{\alpha}$, respectively. 
This can be done by repeating the same construction as in \cite{MM03}. 
\end{proof}


\begin{rem}[Homogeneous rods]
If the rod is made of a homogeneous material, i.e., $W(z,F) = W(F)$, for a.e. $z$ in $\Omega$ and every $F\in \mathbb{M}^{3\times 3}$, then the limiting energy density $Q_2$ is given by the simpler formula
\begin{equation}\label{(a)}
Q_{2}(s,P) = \inf_{\hat{\alpha}\in W^{1,2}(D;\mathbb{R}^{3})}\Bigg\{\int_{D} Q_{3}\bigg(R_{0}(s)\Bigg(P\, \Bigg(\begin{array}{c}
0\\
\xi\\
\zeta
\end{array}\Bigg)\,\bigg|\,\partial_{\xi}\hat{\alpha}\,\bigg|\,\partial_{\zeta}\hat{\alpha}\Bigg)
R_{0}^{T}(s)\bigg)d\xi\,d\zeta\Bigg\}.
\end{equation}
In other words the optimal choice for $\hat{g}$ in (\ref{Q2}) is $\hat{g} = 0$.
\newline
In order to show this, let $\hat{\alpha}\in W^{1,2}(D;\mathbb{R}^{3})$ and let $\hat{g}\in\mathbb{R}^3$. We 
introduce the function
\begin{equation}\label{alfetta}
\tilde{\alpha}(s,\xi,\zeta):= \hat{\alpha}(s,\xi,\zeta) - \xi \int_{D}\partial_{\xi}\hat{\alpha}\, d\xi\,d\zeta - \zeta \int_{D}\partial_{\zeta}\hat{\alpha}\, d\xi\,d\zeta.
\end{equation}
Then,
\begin{align*}
R_{0}\,\Bigg(P\, \Bigg(\begin{array}{c}
0\\
\xi\\
\zeta
\end{array}\Bigg) + \hat{g}\,\bigg|\,\partial_{\xi}\hat{\alpha}\,\bigg|\,\partial_{\zeta}\hat{\alpha}\Bigg)\,
R_{0}^{T} 
=&\, R_{0}\,\Bigg(P\, \Bigg(\begin{array}{c}
0\\
\xi\\
\zeta
\end{array}\Bigg)\,\bigg|\,\partial_{\xi}\tilde{\alpha}\,\bigg|\,\partial_{\zeta}\tilde{\alpha}\Bigg)\,
R_{0}^{T}\\ +& \,R_0\,\left(\hat{g}\,\Big|\,\int_{D}\partial_{\xi}\hat{\alpha}\, d\xi\,d\zeta\,\bigg|\,\int_{D}\partial_{\zeta}\hat{\alpha}\, d\xi\,d\zeta\right)\,R_{0}^{T}\\
=:&\,\, \tilde{G} + Z.
\end{align*}
By expanding the quadratic form $Q_{3}$, we have
\begin{equation}\label{Q3}
\int_{D} Q_{3}(G)d\xi\,d\zeta = \int_{D} Q_{3}(\tilde{G})d\xi\,d\zeta + \int_{D} Q_{3}(Z)d\xi\,d\zeta 
\geq \int_{D} Q_{3}(\tilde{G})d\xi\,d\zeta, 
\end{equation}
where we used (\ref{dom2}), the fact that $\partial_{\xi}\tilde{\alpha}$ and $\partial_{\zeta}\tilde{\alpha}$ 
have zero average on $D$ and the non negativity of $Q_{3}$. 
From this inequality the thesis follows immediately. \\
Notice that, due to the nontrivial geometry of the body, the limit energy depends on the position over the curve $\gamma$ even for a homogeneous material.
\end{rem}


\begin{rem}[Homogeneous and isotropic rods]
Assume the density $W$ is homogeneous and isotropic, that is,
\begin{equation*}
W(F) = W(FR) \quad\mbox{for every} \,\, R\in SO(3).
\end{equation*}
Then the quadratic form $Q_{3}$ is given by
\begin{equation*}
Q_{3}(G) = 2\,\mu\,\bigg|\frac{G + G^{T}}{2}\bigg|^{2} + \lambda\,(\mbox{tr}\, G)^{2}
\end{equation*}
for some constants $\lambda,\mu \,\in \mathbb{R}$. It is easy to show that for all 
$G \in\mathbb{M}^{3\times 3}$ and $R\in SO(3)$
\begin{equation*}
Q_{3}(R\,G\,R^{T}) = Q_{3}(G),
\end{equation*}
and so, formula (\ref{(a)}) reduces to 
\begin{align*}
Q_{2}(P) =&\, \inf_{\hat{\alpha}\in W^{1,2}(D;\mathbb{R}^{3})}\Bigg\{\int_{D} Q_{3} \Bigg(P\, \Bigg(\begin{array}{c}
0\\
\xi\\
\zeta
\end{array}\Bigg)\,\bigg|\,\partial_{\xi}\hat{\alpha}\,\bigg|\,\partial_{\zeta}\hat{\alpha}\Bigg)\,d\xi\,d\zeta\Bigg\}\\
=&\, \frac{1}{2\,\pi}\,\frac{\mu(3\,\lambda + 2\,\mu)}{\lambda + \mu}\,(p_{12}^2 + p_{13}^2) + \frac{\mu}{2\,\pi}\,p_{23}^2,
\end{align*}
where the last equality follows from \cite[Remark 3.5]{MM03}. 
This means that in the case of a homogeneous and isotropic material the quadratic form $Q_2$
is exactly the same as in the case of a straight rod treated in \cite{MM03}. 
\end{rem}


\begin{rem}[Homogeneous rods with a circular cross section]
Assume that the cross section $D$ is a circle of radius $\frac{1}{\sqrt{\pi}}$ centred at the origin.
In this case, the quadratic form $Q_2$ can be computed by a pointwise minimization. 
More precisely, for every $s$ and for every $P$,
\begin{equation*}
Q_{2}(s,P) = \frac{1}{4\pi}\min_{u,v,w}\left\{Q_{3}\Bigg(R_{0}(s)\,\Bigg(\begin{array}{c}
p_{12} \\
0\\
- p_{23}
\end{array}\,\Bigg|\,u\,
\bigg|\,v \Bigg) R_{0}^{T}(s)\Bigg) + Q_{3}\Bigg(R_{0}(s)\,\Bigg(\begin{array}{c}
p_{13}\\
p_{23}\\
0
\end{array}\,\Bigg|\, v\,
\bigg|\, w\Bigg) R_{0}^{T}(s)\Bigg)\right\}.
\end{equation*}
The proof is completely analogous to \cite[Remark 3.6]{MM03}.
\end{rem}


\section{Lower scalings of the energy}
\noindent
The content of this section is the study of the asymptotic behaviour of the functionals 
$\tilde{I}^{(h)}/h^\alpha$\, for \,$0\leq\alpha<2$, as $h\rightarrow 0$.
In addition to conditions (i)-(v) of Section $2$ we assume also that 
$W(z,F) = W(z_{1},F)$ for every $z=(z_{1},z_{2},z_{3})\in\mathbb{R}^{3}$ and every 
$F\in \mathbb{M}^{3\times 3}$, and that  
\begin{align*}
(\textnormal{vi})& \,\, \exists \, C_3 > 0 \,\, \mbox{ independent of $z_1$ such that }\,
W(z_{1},F) \leq \,C_3\,\mbox{dist}^{2}(F,SO(3)) \,\, \mbox{for a.e. $z_1$ } \\
&\mbox{ and every} \, F\in\mathbb{M}^{3\times 3}.
\end{align*}
It is convenient to write the functionals $\tilde{I}^{(h)}$ as integrals over 
the fixed domain $\Omega = \big(\Psi^{(h)}\big)^{-1}\big(\tilde{\Omega}_{h}\big)$. Changing variables as in (\ref{flower2}) and setting $v:= \tilde{v}\circ \Psi^{(h)}$, we have
\begin{equation*}
\tilde{I}^{(h)}(\tilde{v}) = \int_{\Omega} W\big(s,\big(\nabla_h v\big)\,\big(\nabla_h\Psi^{(h)}\big)^{-1}\big)\det
\big(\nabla_{h}\Psi^{(h)}\big)
ds\,d\xi\,d\zeta =: \tilde{J}^{(h)}(v).
\end{equation*}
We extend the functional to the space $L^{2}(\Omega;\mathbb{R}^{3})$, setting
\begin{equation*}
J^{(h)}(v) =
\left\{
\begin{array}{ll}
\vspace{.1cm}
\tilde{J}^{(h)}(v) & \mbox{if }  v\in W^{1,2}(\Omega;\mathbb{R}^{3}),\\
\displaystyle + \infty & \mbox{otherwise in} \, L^{2}(\Omega;\mathbb{R}^{3}).
\end{array}
\right.
\end{equation*}
The aim of this section is to determine the $\Gamma$-limit of 
\,$J^{(h)}/h^\alpha$,\,for \,$0\leq\alpha<2$, as 
$h\rightarrow 0$, with respect to the strong topology of $L^{2}$. 

\subsection{Derivation of the nonlinear theory for curved strings}
\noindent
For this first part we specify $\alpha = 0$, so we are interested in the asymptotic behaviour 
of the functionals representing the energy per unit volume associated to a deformation of the 
reference configuration.
\begin{thm}[Compactness]\label{comp1}
For every sequence $\big(v^{(h)}\big)$ in
$L^{2}(\Omega;\mathbb{R}^{3})$ such
that
\begin{equation}\label{fin}
J^{(h)}\big(v^{(h)}\big) \leq c < +\infty
\end{equation}
there exist a function $v\in W^{1,2}((0, L);\mathbb{R}^{3})$ and some
constants $c^{(h)}\in\mathbb{R}^3$ such that, up to subsequences,
\begin{equation*}
v^{(h)} - c^{(h)} \rightharpoonup v \quad \mbox{weakly in }\,
W^{1,2}(\Omega;\mathbb{R}^{3}). 
\end{equation*}
\end{thm}
\begin{proof}
Let $\big(v^{(h)}\big)$ be a sequence in $L^{2}(\Omega;\mathbb{R}^{3})$ 
satisfying (\ref{fin}). From the definition of the functional we have immediately 
that $v^{(h)}\in W^{1,2}(\Omega;\mathbb{R}^{3})$. 
The assumptions on $W$ and the uniform boundedness of $\big(\nabla_{h}\Psi^{(h)}\big)^{-1}$ and of $\det \big(\nabla_{h}\Psi^{(h)}\big)$ give the boundedness in $L^{2}(\Omega;\mathbb{M}^{3\times 3})$ of $\big(\nabla_{h}v^{(h)}\big)$ and hence of $\big(\nabla v^{(h)}\big)$.
Therefore, using the Poincar\'e inequality
\begin{equation*}
\big|\big|v^{(h)} - c^{(h)}\big|\big|_{L^{2}(\Omega;\mathbb{R}^{3})} \leq \big|\big|\nabla v^{(h)}\big|\big|_{L^{2}(\Omega;\mathbb{M}^{3\times 3})},
\end{equation*}
where $c^{(h)}\in\mathbb{R}^3$ is the mean value of $v^{(h)}$ over $\Omega$,
it turns out that the sequence $v^{(h)} - c^{(h)}$ is bounded in $W^{1,2}(\Omega;\mathbb{R}^{3})$; hence there exists a function $v \in W^{1,2}(\Omega;\mathbb{R}^{3})$ such that, up to subsequences,
$$v^{(h)} - c^{(h)}\rightharpoonup v \quad\mbox{weakly in}\,W^{1,2}(\Omega;\mathbb{R}^{3}).$$
Moreover since $\big(\nabla_{h} v^{(h)}\big)$ is bounded in $L^{2}(\Omega;\mathbb{M}^{3\times 3})$, we have
$$\partial_{\xi} v^{(h)} \rightarrow 0\quad \mbox{and}\quad  \partial_{\zeta} v^{(h)} \rightarrow 0 \quad\mbox{strongly in}\,L^{2}(\Omega;\mathbb{R}^{3}).$$
Therefore the limit function $v$ depends only on the first variable. 
\end{proof}
\begin{thm}[$\Gamma$-convergence]\label{Gamcon}
Let $I$ be the functional defined as 
\begin{equation}\label{Glim}
I(v) =
\left\{
\vspace{.5cm}
\begin{array}{ll}
\displaystyle \int_{0}^{L} W_{0}^{**}(s,v'(s))\,ds & \mbox{if } \, v\in W^{1,2}((0,L);\mathbb{R}^3),\\
\displaystyle + \infty & \mbox{otherwise in }\, L^{2}(\Omega;\mathbb{R}^3), 
\end{array}
\right.
\end{equation}
where $W_{0}^{**}$ is given by the convex envelope of the function
$W_{0}: [0, L]\times \mathbb{R}^{3}\rightarrow \mathbb{R}$ defined as 
$$W_0(s,z):= \inf \big\{ W\left(s,(z\,|\,y_{2}\,|\,y_{3}) R_{0}^{T}(s)\right): y_{2},y_{3} \in \mathbb{R}^{3}\big\}.$$
Then 
\begin{equation*}
\Gamma-\lim_{h\rightarrow 0} J^{(h)} = I, 
\end{equation*}
i.e., the following conditions are satisfied:\\
(i)(liminf inequality) for every $v\in L^{2}(\Omega;\mathbb{R}^{3})$ and every sequence 
$\big(v^{(h)}\big)\subset L^{2}(\Omega;\mathbb{R}^{3})$ such that 
$v^{(h)}  \rightarrow v$ strongly in $L^{2}(\Omega;\mathbb{R}^{3})$, it turns out that
\begin{equation}\label{linf}
I(v)\leq\liminf_{h \rightarrow 0} J^{(h)}\big(v^{(h)}\big);
\end{equation}
(ii)(limsup inequality) for every $v\in L^{2}(\Omega;\mathbb{R}^{3})$ there exists a sequence 
$\big(v^{(h)}\big)\subset L^{2}(\Omega;\mathbb{R}^{3})$ converging strongly to $v$ 
in $L^{2}(\Omega;\mathbb{R}^{3})$ such that
\begin{equation}\label{55}
\limsup_{h \rightarrow 0}
J^{(h)}\big(v^{(h)}\big)\leq I(v).
\end{equation}
\end{thm}
\begin{rem}
Notice that, if $A:= (z\,|\,y_2\,|\,y_3)\,R_0^T$, then $A\,\tau = z$ and $A\,\nu_k = y_k$ for $k = 2,3$.
In other words, in the definition of $W_0$, the minimization is done with respect to the normal components 
of the matrix in the argument of $W$, keeping equal to $z$ the tangential component. 
\end{rem}
\begin{rem}\label{abp}
Observe that conditions (iv) and (v) imply that for a.e. $s\in [0, L]$,
\begin{equation}\label{ABP}
W_{0}^{**}(s,z) = 0 \quad \mbox{if and only if} \quad \snorm{z}\leq 1,
\end{equation}
(see \cite{ABP91}). 
\end{rem}
\begin{proof}(of Theorem \ref{Gamcon})
(i) Let $v$ and $v^{(h)}$ be as in the statement. We can assume that 
\begin{equation*}
\liminf_{h \rightarrow 0} J^{(h)}\big(v^{(h)}\big) < + \infty, 
\end{equation*}
otherwise (\ref{linf}) is trivial. Therefore, up to subsequences, (\ref{fin}) is satisfied. 
From Theorem \ref{comp1} we deduce that $v \in W^{1,2}((0,L);\mathbb{R}^3)$ and that the convergence 
is indeed weak in $W^{1,2}(\Omega;\mathbb{R}^{3})$.\\
Now define the function $W_{0}: [0, L]\times \mathbb{R}^{3}\rightarrow \mathbb{R}$ as
$$W_{0}(s,z):= \inf \big\{W\left(s,(z\,|\,y_{2}\,|\,y_{3}) R_{0}^{T}(s)\right): y_{2},y_{3} \in \mathbb{R}^{3}\big\}.$$
Due to the coercivity assumptions this function is finite.

Notice that, since $R_{0}\,R_{0}^{T} = Id$, we can write
\begin{equation*}
W\big(s,\nabla_{h}v^{(h)}\big(\nabla_{h}\Psi^{(h)}\big)^{-1}\big) =
W\Big(s,\nabla_{h}v^{(h)}\big(\nabla_{h}\Psi^{(h)}\big)^{-1}R_{0}R_{0}^{T}\Big)
\end{equation*}
and using the explicit expression of $\big(\nabla_{h}\Psi^{(h)}\big)^{-1}$ given in (\ref{invA}), i.e.,
\begin{equation*}
\big(\nabla_{h}\Psi^{(h)}\big)^{-1}(s,\xi,\zeta) =
R_{0}^{T}(s) - h\,R_{0}^{T}(s)\,\big[(\xi\,\nu'_{2}(s) +
\zeta\,\nu'_{3}(s))\otimes e_1\big] R_{0}^{T}(s) + O(h^{2}),
\end{equation*}
we have
\begin{equation}\label{convVp}
\nabla_{h} v^{(h)}\big(\nabla_{h}\Psi^{(h)}\big)^{-1}R_{0}e_1 \rightharpoonup v' \quad \mbox{weakly in } \, L^{2}(\Omega;\mathbb{R}^3).
\end{equation}
So, from the definition of $W_{0}$
\begin{eqnarray*}
J^{(h)}\big(v^{(h)}\big)  &\geq&
\int_{\Omega} W_{0}\big(s,\nabla_{h} v^{(h)}\big(\nabla_{h}\Psi^{(h)}\big)^{-1}R_{0}e_1\big)\,\det \big(\nabla_{h}\Psi^{(h)}\big)\,ds\,d\xi\,d\zeta \\
&\geq& \int_{\Omega} W_{0}^{**}\big(s,\nabla_{h} v^{(h)}\big(\nabla_{h}\Psi^{(h)}\big)^{-1}R_{0}e_1\big)\,\det \big(\nabla_{h}\Psi^{(h)}\big)\,ds\,d\xi\,d\zeta.
\end{eqnarray*}
Now we pass to the $\liminf$ in both sides of the previous inequality, using the uniform 
convergence of the determinant remarked in (\ref{convdet}), and we get
\begin{eqnarray*}
\liminf_{h\rightarrow 0}\,J^{(h)}\big(v^{(h)}\big)&\geq&
\liminf_{h\rightarrow 0}\int_{\Omega} W_{0}^{**}\big(s,\big(\nabla_{h} v^{(h)}\big(\nabla_{h}\Psi^{(h)}\big)^{-1}R_{0}\big)e_1\big)\,\det \big(\nabla_{h}\Psi^{(h)}\big)\,ds\,d\xi\,d\zeta \\
&=& \liminf_{h\rightarrow 0}\int_{\Omega}W_{0}^{**}\big(s,\big(\nabla_{h} v^{(h)}\big(\nabla_{h}\Psi^{(h)}\big)^{-1}R_{0}\big)e_1\big)\,ds\,d\xi\,d\zeta.
\end{eqnarray*}
Since the functional
$$G(u):= \int_{\Omega} W_{0}^{**}(s, u)\,ds\,d\xi\,d\zeta$$
is convex, it is sequentially weakly lower semicontinuous in $L^{2}(\Omega;\mathbb{R}^{3})$; 
so, by (\ref{convVp}) we can conclude that
\begin{equation}\label{supe1}
\liminf_{h\rightarrow 0}\,J^{(h)}\big(v^{(h)}\big) \geq \int_{0}^{L} W_{0}^{**}(s,v'(s))\,ds.
\end{equation}
(ii) Let $v$ be a function in $W^{1,2}((0, L);\mathbb{R}^3)$, otherwise the bound in (\ref{55}) is trivial. 
Let $w_{2}, w_{3} \in W^{1,2}((0, L);\mathbb{R}^{3})$ be arbitrary functions and consider the functions 
$v^{(h)}: \Omega\rightarrow \mathbb{R}^{3}$ defined by
\begin{equation*}
v^{(h)}(s,\xi,\zeta) := v(s) + h\,\xi\,w_{2}(s) + h\,\zeta\,w_{3}(s).
\end{equation*}
Clearly, as $\nabla v^{(h)} = v'\otimes e_1 + h\,\big(\xi\,w'_2 + \zeta\,w'_3\,|\,w_2\,|\,w_3\big)$, we have that
\begin{equation}\label{(b)}
v^{(h)}\rightarrow v\quad \mbox{strongly in}\, W^{1,2}(\Omega;\mathbb{R}^{3}).
\end{equation}
Now we want to study the behaviour of the sequence 
\begin{equation*}
J^{(h)}\big(v^{(h)}\big) = \int_{\Omega}W\big(s,(\nabla_{h} v^{(h)})\big(\nabla_{h}\Psi^{(h)}\big)^{-1}) \det \big(\nabla_{h}\Psi^{(h)}\big) ds\,d\xi\,d\zeta
\end{equation*}
when $h\rightarrow 0$. Notice that the scaled gradient of $v^{(h)}$ satisfies 
\begin{equation}\label{almever}
\nabla_{h} v^{(h)} = (v'\,|\,w_2\,|\,w_3) + 
h\, (\xi\,w'_2 + \zeta\, w'_3)\otimes e_1 \rightarrow (v'\,|\,w_2\,|\,w_3) \,\, \mbox{a.e.}.
\end{equation}
So, by (\ref{convdet}) and (vi), using the dominated convergence theorem we get
\begin{align*}
\lim_{h\rightarrow 0}J^{(h)}\big(v^{(h)}\big) &=\, \lim_{h\rightarrow 0}\int_{\Omega} W\big(s,(\partial_{s} v^{(h)}\,|\,w_{2}\,|\,w_{3})\big(\nabla_{h}\Psi^{(h)}\big)^{-1}) \det \big(\nabla_{h}\Psi^{(h)}\big)\, ds\,d\xi\,d\zeta\\
&=\, \int_{0}^{L} W\big(s,(v'\,|\,w_{2}\,|\,w_{3})\,R_{0}^{T})\, ds.
\end{align*}
Up to now we have shown that for every choice of $w_{2}, w_{3} \in W^{1,2}((0, L);\mathbb{R}^{3})$, there exists a sequence $\big(v^{(h)}\big)$ such that (\ref{(b)}) is satisfied and 
\begin{equation*}
\lim_{h\rightarrow 0}J^{(h)}\big(v^{(h)}\big)  = \int_{0}^{L} W\big(s,(v'\,|\,w_{2}\,|\,w_{3})\,R_{0}^{T}) ds.
\end{equation*}
Therefore,
\begin{align}\label{densi}
\Gamma-\limsup_{h\rightarrow 0} J^{(h)}(v)&:=\,
\inf\left\{\limsup_{h\rightarrow 0}J^{(h)}\big(u^{(h)}\big):
u^{(h)} \rightarrow v \,\, \mbox{strongly in} \,\, L^{2}(\Omega;\mathbb{R}^{3})\right\}\nonumber\\ 
&\leq \inf\left\{\int_{0}^{L} W\big(s,(v'\,|\,w_{2}\,|\,w_{3})\,R_{0}^{T})\,ds :
w_{2}, w_{3} \in W^{1,2}((0, L);\mathbb{R}^{3})\right\}\nonumber\\
&= \inf\left\{\int_{0}^{L} W\big(s,(v'\,|\,w_{2}\,|\,w_{3})\,R_{0}^{T})\,ds :
w_{2}, w_{3} \in L^{2}((0, L);\mathbb{R}^{3})\right\},
\end{align}
where the last equality is a consequence of the dominated convergence theorem and of the density of $W^{1,2}((0, L);\mathbb{R}^{3})$ in $L^{2}((0, L);\mathbb{R}^{3})$. 

By the measurable selection lemma (see for example \cite{EkTe}) applied to the Carath\'eodory function
$$g:[0, L]\times\mathbb{R}^{3}\times\mathbb{R}^{3}\rightarrow \mathbb{R},\quad (s,y_{2},y_{3})\mapsto g(s,y_{2},y_{3}):=  W\big(s,(v'(s)\,|\,y_{2}\,|\,y_{3})R_{0}^{T}(s))$$
we obtain the existence of two measurable functions $w^{0}_{2}, w^{0}_{3}: [0, L]\rightarrow \mathbb{R}^{3}$ satisfying
\begin{equation*}
W\big(s,(v'(s)\,|\,w^{0}_{2}(s)\,|\,w^{0}_{3}(s))R_{0}^{T}(s)) =\inf_{y_{2},y_{3}\in \mathbb{R}^{3}} W\big(s,(v'(s)\,|\,y_{2}\,|\,y_{3})R_{0}^{T}(s)) = W_{0}(s,v'(s)).
\end{equation*}
Moreover, from the coerciveness of $W$ it follows that $w^{0}_{2}, w^{0}_{3}$ belong indeed to $L^{2}((0, L);\mathbb{R}^{3})$ and so they are in competition for the infimum in (\ref{densi}). 
Hence, for every $v\in W^{1,2}((0, L);\mathbb{R}^{3})$ we have
\begin{equation*}
\Gamma-\limsup_{h\rightarrow 0} J^{(h)}(v) \leq \int_{0}^{L} W_{0}(s,v'(s))\,ds =: \tilde{J}(v).
\end{equation*}
Now define the functional
\begin{equation}\label{Glim2}
J(v) =
\left\{
\vspace{.5cm}
\begin{array}{ll}
\displaystyle \tilde{J}(v) & \mbox{if } \, v\in W^{1,2}((0,L);\mathbb{R}^3),\\
\displaystyle + \infty & \mbox{otherwise in }\, L^{2}(\Omega;\mathbb{R}^3);
\end{array}
\right.
\end{equation}
clearly it turns out that 
\begin{equation}\label{darelax} 
\Gamma-\limsup_{h\rightarrow 0} J^{(h)}(v) \leq J(v) \quad\mbox{for every} \,\, v\in L^{2}(\Omega;\mathbb{R}^{3}).
\end{equation} 
As the lower semicontinuous envelope of $J$ with respect to the strong topology of $L^{2}(\Omega;\mathbb{R}^{3})$ is given by the functional $I$ (see \cite{DM93} and \cite[Lemma 5]{LDR95}), the thesis follows immediately from (\ref{darelax}).
\end{proof}


\subsection{Intermediate scaling }
\noindent
In this subsection we show that scalings of the energy of order $h^\alpha$, with $\alpha\in (0,2)$, lead to a trivial $\Gamma$-limit.

\begin{thm}[Compactness and $\Gamma$- convergence]
Let $\mathcal{W}_1$ be the class of functions defined as
\begin{equation}\label{defW1M}
\mathcal{W}_1:= \{v\in W^{1,2}((0,L);\mathbb{R}^{3}) : |v'(s)| \leq 1 \,\textnormal{a.e.}\}.
\end{equation}
For every sequence $\big(v^{(h)}\big)$ in
$L^{2}(\Omega;\mathbb{R}^{3})$ such
that
\begin{equation}\label{fin2}
\frac{1}{h^\alpha}\,J^{(h)}\big(v^{(h)}\big) \leq c < +\infty
\end{equation}
there exist a function $v\in \mathcal{W}_{1}$ and some
constants $c^{(h)}\in\mathbb{R}$ such that, up to subsequences,
\begin{equation*}
v^{(h)} - c^{(h)} \rightharpoonup v \quad \mbox{weakly in }\,
W^{1,2}(\Omega;\mathbb{R}^{3}). 
\end{equation*}
Moreover, 
\begin{equation}\label{(c)}
\Gamma-\lim_{h \rightarrow 0}\, \frac{1}{h^\alpha}\,J^{(h)} =
\left\{
\vspace{.3cm}
\begin{array}{ll}
\vspace{.15cm}
\quad 0 &  \textnormal{in} \,\, \mathcal{W}_1,\\
\displaystyle + \infty & \textnormal{otherwise in} \,\, L^{2}(\Omega;\mathbb{R}^{3}).
\end{array}
\right.
\end{equation}
\end{thm}
\begin{proof}Let $\big(v^{(h)}\big)$ be such that (\ref{fin2}) is satisfied. Then
\begin{equation}\label{boundalfa}
J^{(h)}\big(v^{(h)}\big) < c\,h^{\alpha}.
\end{equation}
By Theorem \ref{comp1} this implies that there exist $v\in W^{1,2}((0,L);\mathbb{R}^{3})$ and some constants $c^{(h)}\in\mathbb{R}$ such that the sequence $v^{(h)} - c^{(h)}$ converges to $v$ weakly in $W^{1,2}(\Omega;\mathbb{R}^{3})$. Moreover by Theorem \ref{Gamcon} and by (\ref{boundalfa})
\begin{equation*}
0 = \liminf_{h\rightarrow 0}J^{(h)}\big(v^{(h)}\big)  \geq \int_{0}^{L} W_{0}^{**}(s,v'(s))ds,
\end{equation*}
and this gives the additional condition that $\snorm{v'(s)}\leq 1$ for almost every $s\in [0,L]$, thanks to Remark \ref{abp}. Therefore $v\in \mathcal{W}_{1}$.\\
Let us prove (\ref{(c)}). The liminf inequality follows directly from the fact that the energy density $W$ is nonnegative and from the compactness. As for the limsup inequality we first notice that we can restrict our analysis to functions $v\in \mathcal{W}_{1}$, being the other case trivial. Since $\snorm{v'(s)}\leq 1$ for a.e. $s\in [0, L]$, there exist two measurable functions $d_{2}, d_{3}: [0, L]\rightarrow \mathbb{R}^3$ such that
\begin{equation*}
(v'(s)\,|\,d_{2}(s)\,|\,d_{3}(s))\in Co(SO(3))\quad \mbox{for a.e.}\,\, s\in [0,L],
\end{equation*}
where $Co(SO(3))$ denotes the convex hull of $SO(3).$ As first step, we assume in addition that $(v'\,|\,d_{2}\,|\,d_{3})$ is a piecewise constant rotation; for simplicity we can limit ourselves to the case
$$
(v'(s)\,|\,d_{2}(s)\,|\,d_{3}(s)) =
\left\{
\begin{array}{ll}
\vspace{.15cm}
R_{1} & \mbox{if }  s\in [0,s_{0}[,\\
R_{2} & \mbox{if }  s\in [s_{0},L]
\end{array}
\right.
$$
with $R_{1},R_{2} \in SO(3)$.
Now, let $\omega(h)$ be a sequence converging to zero, as $h\rightarrow 0$, and let $P$ be a smooth function \,$P:[0,1] \longrightarrow SO(3)$, such that $P(0) = R_{1}$ and $P(1) = R_{2}$.
Now consider a reparametrization of $P$, denoted by $P^{(h)}$ and given by
\begin{equation*}
P^{(h)}(s) := P\bigg(\frac{s - s_{0}}{\omega(h)}\bigg).
\end{equation*}
Define the sequence $v^{(h)}: \Omega\rightarrow \mathbb{R}^3$ as
$$
v^{(h)}(s,\xi,\zeta) :=
\left\{
\begin{array}{lll}
R_{1}\Biggl(\begin{array}{c}
s\\
h\,\xi\\
h\,\zeta
\end{array}\Biggr)  & \mbox{on }  s\in [0,s_{0}[\times D,\\
\displaystyle\int_{s_{0}}^{s} \big(P^{(h)}\big)(\sigma)e_1\,d\sigma +  P^{(h)}(s)\Biggl(\begin{array}{c}
0\\
h\,\xi\\
h\,\zeta
\end{array}\Biggr) + b^{(h)}    & \mbox{on } \, \big[s_{0}, s_{0} + \omega(h)\big]\times D,\\
R_{2}\Biggl(\begin{array}{c}
s\\
h\,\xi\\
h\,\zeta
\end{array}\Biggr) + d^{(h)} & \mbox{on }  \, \big]s_{0} + \omega(h), L\big]\times D,\\
\end{array}
\right.
$$
where the constants $b^{(h)}$ and $d^{(h)}$ are chosen in order to make $v^{(h)}$ continuous.
It turns out that the scaled gradient has the following expression:
\begin{equation}\label{Gradve}
\nabla_{h}v^{(h)} = \left\{
\begin{array}{lll}
R_{1} & \mbox{on } \,  [0,s_{0}[\times D,\\
P^{(h)}(s) +  \Bigg(\big(P^{(h)}\big)'(s)\Bigg(\begin{array}{c}
0\\
h\,\xi\\
h\,\zeta
\end{array}\Bigg)\Bigg)\otimes e_1  & \mbox{on } \, \big[s_{0}, s_{0} + \omega(h)\big]\times D,\\
R_{2} & \mbox{on }  \, \big]s_{0} + \omega(h), L\big]\times D;\\
\end{array}
\right.
\end{equation}
moreover $\nabla_{h}v^{(h)}\rightarrow (v'\,|\,d_{2}\,|\,d_{3})$ strongly in $L^{2}(\Omega;\mathbb{R}^{3})$. In order to evaluate the functional on this sequence we use the fact that, by (v) and (\ref{convdet}),
\begin{equation}\label{estimdist}
\frac{1}{h^\alpha}\,J^{(h)}\big(v^{(h)}\big) \leq \frac{c}{h^\alpha}\, \int_{\Omega}\mbox{dist}^2\big(\nabla_{h}v^{(h)}\,
\big(\nabla_{h}\Psi^{(h)}\big)^{-1}, SO(3)\big)\,ds\,d\xi\,d\zeta.
\end{equation}
From (\ref{Gradve}) the integral on the right-hand side of the previous expression can be written as
\begin{align}\label{3int}
&\int_{0}^{s_0}\int_{D}\mbox{dist}^2\big(R_1\,
\big(\nabla_{h}\Psi^{(h)}\big)^{-1}, SO(3)\big)\,ds\,d\xi\,d\zeta +\,\int_{s_0 + \omega(h)}^{L}\int_{D}\mbox{dist}^2\big(R_2\,
\big(\nabla_{h}\Psi^{(h)}\big)^{-1}, SO(3)\big)\,ds\,d\xi\,d\zeta \nonumber\\ 
&+\int_{s_0}^{s_0 + \omega(h)}\int_{D}\mbox{dist}^2\big(\nabla_{h}v^{(h)}\,
\big(\nabla_{h}\Psi^{(h)}\big)^{-1}, SO(3)\big)\,ds\,d\xi\,d\zeta.
\end{align}
The first two terms in (\ref{3int}) give a contribution of order $h^2$ since, by (\ref{invA}), for $i= 1,2$,
\begin{align*}
\mbox{dist}^2\big(R_i\,\big(\nabla_{h}\Psi^{(h)}\big)^{-1}, SO(3)\big) &\leq h^2\,  \mbox{dist}^2\big(R_i\,R_{0}^{T}\,\big[(\xi\,\nu'_{2} +
\zeta\,\nu'_{3})\otimes e_1\big] R_{0}^{T}, SO(3)\big)\\
&\leq \,C\,h^2\, \mbox{dist}^2\big(\big[(\xi\,\nu'_{2} +
\zeta\,\nu'_{3})\otimes e_1\big], SO(3)\big),
\end{align*}
so they can be neglected in the computation of the limit of (\ref{estimdist}). The only term we have to analyse is the last integral in (\ref{3int}). Set 
$$A^{(h)}(s,\xi,\zeta):= \Bigg(\big(P^{(h)}\big)'\Bigg(\begin{array}{c}
0\\
h\,\xi\\
h\,\zeta
\end{array}\Bigg)\Bigg)\otimes e_1.$$
Using again (\ref{invA}) we have that
\begin{equation*}
\mbox{dist}^2\big(\nabla_{h}v^{(h)}\,\big(\nabla_{h}\Psi^{(h)}\big)^{-1}, SO(3)\big)\leq \,\mbox{dist}^2\,\big(A^{(h)}\,\big(\nabla_{h}\Psi^{(h)}\big)^{-1}, SO(3)\big)\leq \,C\,h^2\,\big(\xi^{2} + \zeta^{2}\big)\, \big|\,\big(P^{(h)}\big)'\,\big|^2,  
\end{equation*}
so we get the following estimate:
\begin{align*}
\int_{s_0}^{s_0 + \omega(h)}\int_{D}\mbox{dist}^2\big(\nabla_{h}v^{(h)}\,
\big(\nabla_{h}\Psi^{(h)}\big)^{-1}, SO(3)\big)\,ds\,d\xi\,d\zeta &\leq 
C\,h^2\,\int_{s_0}^{s_0 + \omega(h)}\,\big|\,\big(P^{(h)}\big)'\,\big|^2\,ds\\
&= \,C \,\frac{h^{2}}{\omega(h)}\int_{0}^{1} \big|\,P'\big|^{2} ds.
\end{align*}
Notice that, if we choose $\omega(h)\sim h^{\beta}$, with $0<\beta<2 - \alpha$, also this term can be neglected in (\ref{estimdist}), hence
\begin{equation*}
\lim_{h\rightarrow 0}\frac{1}{h^\alpha}\,J^{(h)}\big(v^{(h)}\big) = 0
\end{equation*}
and this concludes the proof in the case $(v'\,|\,d_{2}\,|\,d_{3})$ is a piecewise constant rotation. 

Consider now the general case. Since $(v'\,|\,d_{2}\,|\,d_{3}) \in Co(SO(3))$ \,a.e., there exists a sequence of piecewise constant rotations $R_j: [0, L]\longrightarrow SO(3)$ such that
$R_j \rightarrow (v'\,|\,d_{2}\,|\,d_{3})$ strongly in $L^{2}((0, L);\mathbb{M}^{3\times 3})$.
For each element $R_j$ of the sequence we can repeat the same construction done in the previous case and find a sequence 
$v^{(h)}_j$ whose scaled gradients $\nabla_{h}v^{(h)}_j$ converge to $R_j$ as $h\rightarrow 0$ and such that for every $j$
\begin{equation}
\lim_{h\rightarrow 0}\frac{1}{h^{\alpha}}\int_{\Omega}W(s,\nabla_{h} v_{j}^{(h)}\, \big(\nabla_{h}\Psi^{(h)}\big)^{-1}\big)\det \big(\nabla_{h}\Psi^{(h)}\big) ds\,d\xi\,d\zeta = 0.
\end{equation}
Now we can choose, for every $j$, an element of the sequence $v^{(h)}_j$, say $v^{(h_j)}_j$, in such a way that 
\begin{equation}
\norm{\nabla_{h_j}v^{(h_j)}_j - R_{j}}_{L^2(\Omega;\mathbb{M}^{3\times 3})} < \frac{1}{j}
\end{equation}
and
\begin{equation}
\frac{1}{h_j^{\alpha}}\int_{\Omega}W(s,\nabla_{h_j} v_{j}^{(h_j)}\, \big(\nabla_{h_j}\Psi^{(h_j)}\big)^{-1}\big)\det \big(\nabla_{h_j}\Psi^{(h_j)}\big) ds\,d\xi\,d\zeta < \frac{1}{j}.
\end{equation}
These estimates show that the sequence $v^{(h_j)}_j$ converges to $(v'\,|\,d_{2}\,|\,d_{3})$ strongly in 
$L^{2}((0, L);\mathbb{M}^{3\times 3})$ and that 
\begin{equation}
\lim_{j\rightarrow \infty}\frac{1}{h_j^{\alpha}}\int_{\Omega}W(s,\nabla_{h_j} v^{(h_j)}\, \big(\nabla_{h_j}\Psi^{(h_j)}\big)^{-1}\big)\det \big(\nabla_{h_j}\Psi^{(h_j)}\big) ds\,d\xi\,d\zeta = 0.
\end{equation}
This concludes the proof.
\end{proof}

\bigskip
\bigskip
\centerline{\textsc{Acknowledgments}}
\bigskip
\noindent 
I would like to thank Maria Giovanna Mora for having proposed
to me the study of this problem and for many helpful and interesting
suggestions. I would like also to thank Gianni Dal Maso for several
stimulating discussions on the subject of this paper. 

\noindent 
This work is part of the project ``Calculus of Variations" 2004, 
supported by the Italian Ministry of Education, University, and Research.
\bigskip


\end{document}